\documentclass[a4paper,12pt]{article}
\usepackage{amsmath,amsthm,amssymb,bbm}

\usepackage[english]{babel}
\newtheorem{theorem}{Theorem}
\newtheorem{corollary}{Corollary}
\newtheorem{lemma}{Lemma}
\newtheorem{definition}{Definition}
\newtheorem{proposition}{Proposition}

\title{\bf Entanglement in Algebraic Quantum Mechanics:
Majorana fermion systems}

\author{F. Benatti$^{a,b}$, 
R. Floreanini$^{b}$\\
\\
\small ${}^a$Dipartimento di Fisica, Universit\`a di Trieste, 
34151 Trieste, Italy\\
\small ${}^b$Istituto Nazionale di Fisica Nucleare, Sezione di Trieste,
34151 Trieste, Italy
}

\date{\null}

\begin{document}

\maketitle

\begin{abstract}
\noindent
Many-body entanglement is studied within the algebraic approach to quantum physics
in systems made of Majorana fermions. In this framework, the notion of separability
stems from partitions of the algebra of observables and properties of
the associated correlation functions, rather than on particle tensor products.
This allows obtaining a complete characterization of non-separable Majorana fermion states.
These results may find direct applications in quantum metrology:
using Majorana systems, sub-shot noise accuracy in parameter estimations can be
achieved without preliminary, resource consuming, state entanglement operations.

\end{abstract}

\vskip 1cm

\section{Introduction}

Majorana fermions describe real fermion excitations that can be thought of as half
of normal fermions, in the sense that a complex fermion mode can be obtained by
putting together two real ones. Although originally 
introduced in particle physics \cite{Majorana,Elliott},
as fermions that are their own antifermions, they have found applications in various
branches of physics.%
\footnote{For instance, see the reviews \cite{Elliott}-\cite{Hassler} and references therein.}

Most notably, Majorana fermions appear as quasi-particle excitations in the
so-called topological superconductors \cite{Elliott}-\cite{Stanescu}. 
They turn out to be spatially separated
and this delocalized character protects them from decoherence effects generated
by any local interaction; furthermore, they exhibit non-Abelian statistics.
These two characteristic properties make Majorana excitations in superconductors very
attractive as building blocks for topological quantum computation, where logical qubits
are encoded in states of non-Abelian anyons and logical operations are performed
through their braiding transformations \cite{Nayak}-\cite{Kitaev2}.

Typical resources needed in quantum computation and communication are non classical
correlations and entanglement. While in the case of distinguishable particles
the notion of separability and entanglement is well understood, in systems made
of identical constituents, like Majorana fermions, these notions are still unsettled.%
\footnote{The literature on entanglement
in many-body systems is vast, {\it e.g.} see
\cite{Schliemann}-\cite{Modi}; however, 
only a limited part of the reported results are applicable 
to systems composed by identical constituents.}
The key observation is that the tensor product structure of the multiparticle Hilbert space
is in general no longer available when the particles are indistinguishable. 
As a consequence,
the standard definition of entanglement based on this structure loses its meaning when dealing
with bosons and fermions, or more in general anyons. In such cases, the emphasis should
shift from the system states to the algebra of its observables \cite{Werner1}-\cite{Moriya1}, 
treating the system Hilbert space as an emergent concept.

This change in perspective is dictated by physical considerations: since particles
are identical, they can not be singly addressed nor they individual properties
measured, only collective, global observables being 
in fact experimentally accessible \cite{Feynman,Sakurai}.
In other terms, when dealing with many-body systems, the presence of entanglement
should be identified through the properties of the observable correlation functions
and not by a priori properties of the system states \cite{Zanardi}-\cite{Benatti7}.

This more general approach to separability and entanglement finds its more natural formulation
in the so-called {\it algebraic approach} to quantum physics \cite{Bratteli}-\cite{Strocchi4}.
In this more general framework,
a quantum system is defined through the algebra $\cal A$ containing all its observables.
A state $\Omega$ for the system is just a positive linear map from $\cal A$ 
to the complex numbers, so that, for any observable $\alpha \in {\cal A}$,
{\it i.e.} a hermitian element of $\cal A$,
the real number $\Omega(\alpha)$ gives its experimentally accessible mean value.
From the couple $({\cal A}, \Omega)$, one then deduces through a standard procedure
(the so-called Gelfand-Naimark-Segal construction) the Hilbert space ${\cal H}_\Omega$
containing the system states. In this very general framework, many-body entanglement
can be naturally identified by the presence of nonclassical correlations
among suitable observables.

In the following, we shall apply this general approach to the study of the notions
of separability and entanglement in systems made of Majorana fermions. In this case
the algebra $\cal A$ containing the observables of the systems turns out to be
a Clifford algebra \cite{Gilbert}-\cite{Meinrenken}, 
whose representation theory results quite different from 
the usual Fock representation of the algebra of creation and annihilation operators 
of standard fermions. This poses new questions concerning the relation between 
entanglement and the reducibility of the algebra representation, 
making the theory of Majorana fermion entanglement much richer than 
that of standard fermions or bosons.

These results may have direct applications in quantum technology,
especially in using quantum interferometric devices \cite{Leggett1}-\cite{Yukalov}
to perform metrological tasks. Indeed, as in the case of 
boson and fermion systems \cite{Benatti1,Argentieri,Benatti6},
also for Majorana systems, parameter estimation accuracies going beyond the classical shot-noise
limit can be obtained by exploiting some sort of quantum non-locality embedded in the
measuring apparatus, without the need of preliminary state operations. 
In this respect, Majorana fermions
might turn out to be very useful not only in performing decoherence protected
quantum computational tasks, but also in the development of the next generation
of highly sensitive quantum sensors.

\section{Algebraic approach to Quantum Mechanics}

For completeness, in this Section we shall briefly summarize 
the main features of the algebraic formulation
to quantum physics, underlying the concepts and tools that will be needed
in the following discussions.%
\footnote{For a more detailed discussion, see the reference textbooks 
\cite{Bratteli}-\cite{Strocchi2}.}

Any quantum system can be characterized by the collections of observations
that can be made on it through suitable measurement processes \cite{Strocchi3}. The physical quantities
that are thus accessed are the observables of the system, forming an algebra $\cal A$
under multiplication and linear combinations, the algebra of observables.

\medskip
\noindent
$\bullet$ {\bf $C^\star$-algebras}\hfill\break
In general, $\cal A$ turns out to be a non-commutative $C^\star$-algebra; this means that it is a linear,
associative algebra (with unity) over the field of complex numbers $\mathbb{C}$,
{\it i.e.} a vector space over $\mathbb{C}$, with an associative product, linear
in both factors. Further, $\cal A$ is endowed with an operation of conjugation:
it posses an antilinear involution $\star: {\cal A}\to {\cal A}$, such that
$(\alpha^\star)^\star=\alpha$, for any element $\alpha$ of $\cal A$.
In addition, a norm $|| \cdot ||$ is defined on $\cal A$, satisfying
$||\alpha\beta||\leq ||\alpha||\, ||\beta||$, for any $\alpha,\, \beta \in {\cal A}$
(thus implying that the product operation is continuous), and such that
$||\alpha^\star \alpha||=||\alpha||^2$, so that $||\alpha^\star||=||\alpha||$;
moreover, $\cal A$ is closed under this norm, meaning that $\cal A$ is a complete
space with respect to the topology induced by the norm (a property that in turn makes $\cal A$
a Banach algebra).

In the case of an $n$-level system, $\cal A$ can be identified with the $C^\star$-algebra
${\cal M}_n(\mathbb{C})$ of complex $n\times n$ matrices; the $\star$-operation
coincides now with the hermitian conjugation, $M^\star=M^\dagger$, 
for any element $M\in {\cal M}_n(\mathbb{C})$,
while the norm $||M||$ is given by the largest eigenvalue of $M^\dagger M$.
Nevertheless, the description of a physical system through its $C^\star$-algebra of observables
is particularly appropriate in presence of an infinite number of degrees of freedom,
where the canonical formalism results problematic. Indeed, the algebra ${\cal B}(\cal{H})$
of all bounded operators on an infinite-dimensional Hilbert space $\cal H$ is another
canonical example of a $C^\star$-algebra, when equipped with the usual operator norm
and adjoint operation. Actually, any $C^\star$-algebra is isomorphic to a
norm-closed self-adjoint subalgebra of ${\cal B}(\cal{H})$, for some Hilbert space
$\cal H$ (Gelfand-Naimark theorem).
It is worth adding that the non-commutativity of the algebra of observables $\cal A$ can
be taken as the distinctive property characterizing quantum systems.

\medskip
\noindent
$\bullet$ {\bf States on $C^\star$-algebras}\hfill\break
Although the system observables, {\it i.e.} the hermitian elements of $\cal A$, can be
identified with the physical quantities measured in experiments, the explicit link between
the algebra $\cal A$ and the outcome of the measurements is given by the concept of a state
$\Omega$, through which the expectation value $\Omega(\alpha)$ of the observable $\alpha \in {\cal A}$
can be defined.

In general, a state $\Omega$ on a $C^\star$-algebra $\cal A$ is a linear map
$\Omega: {\cal A} \to \mathbb{C}$, with the property of being positive,
{\it i.e.} $\Omega(\alpha^\star \alpha)\geq 0$, $\forall\alpha \in {\cal A}$,
and normalized, $\Omega(1_{\cal A})=1$, $1_{\cal A}$ being the unit of $\cal A$.
It immediately follows that the map $\Omega$ is also continuous:
$|\Omega(\alpha)|\leq ||\alpha||$, for all $\alpha \in {\cal A}$.

This general definition of state of a quantum system comprises the standard one
in terms of normalized density matrices on a Hilbert space $\cal H$;
indeed, any density matrix $\rho$ defines a state $\Omega_\rho$ on
${\cal B}({\cal H})$ through the relation
\begin{equation}
\Omega_\rho(\alpha)={\rm Tr}[\rho\,\alpha]\ ,\qquad \forall\alpha\in {\cal B}({\cal H})\ ,
\label{2.1}
\end{equation}
which for pure states, $\rho=|\psi\rangle\langle\psi|$, reduces to the
standard expectation: $\Omega_\rho(\alpha)=\langle \psi |\alpha|\psi\rangle$.
Nevertheless, the definition in terms of $\Omega$ is more general, holding
even for systems with infinitely many degrees of freedom, for which
the usual approach in terms of state vectors may be unavailable.

As for density matrices on a Hilbert space $\cal H$, a state $\Omega$ on a $C^\star$-algebra $\cal A$
is said to be pure if it can not be decomposed as a convex sum of two
states, {\it i.e.} if the decomposition $\Omega=\lambda\,\Omega_1+(1-\lambda)\, \Omega_2$,
with $0\leq\lambda\leq 1$, holds only for $\Omega_1=\Omega_2=\Omega$. If a state $\Omega$
is not pure, it is called mixed. It is worth noticing that, for consistency, the assumed
completeness of the relation between observables and measurements on a physical system
requires that the observables separate the states, {\it i.e.} $\Omega_1(\alpha)=\Omega_2(\alpha)$
for all $\alpha \in {\cal A}$ implies $\Omega_1=\Omega_2$, and similarly that the states separate
the observables, {\it i.e.} $\Omega(\alpha)=\Omega(\beta)$ for all states $\Omega$ on $\cal A$
implies $\alpha=\beta$.

\medskip
\noindent
$\bullet$ {\bf GNS-Construction}\hfill\break
Although the above description of a quantum system through its $C^\star$-algebra of observables
(its measurable properties) and states over it (giving the observable expectations)
looks rather abstract, it actually allows an Hilbert space interpretation, through the so-called
{\it Gelfang-Naimark-Segal(GNS)-construction}.

\begin{theorem}
Any state $\Omega$ on the $C^\star$-algebra $\cal A$
uniquely determines (up to isometries) a representation $\pi_\Omega$ of the elements of $\cal A$
as operators in a Hilbert space ${\cal H}_\Omega$, containing a reference vector
$|\Omega\rangle$, whose matrix elements reproduce the observable expectations:
\begin{equation}
\Omega(\alpha)=\langle\Omega| \pi_\Omega(\alpha) |\Omega\rangle\ ,\qquad \alpha\in{\cal A}\ .
\label{2.2}
\end{equation}
\end{theorem}
\begin{proof}
The algebra $\cal A$ can be viewed as a vector space by associating to each element $\alpha\in{\cal A}$
a vector $|\psi_\alpha\rangle$, and (assuming the state $\Omega$ to be faithful,
{\it i.e.} $\Omega(\alpha^\star\alpha)>0$ for all non vanishing $\alpha$) by introducing the positive
definite inner product $\langle\psi_\alpha|\psi_\beta\rangle=\Omega(\alpha^\star\beta)$.
The completion of $\cal A$ in the corresponding norm gives an Hilbert space ${\cal H}_\Omega$.
The representation $\pi_\Omega: {\cal A} \to {\cal B}({\cal H}_\Omega)$ 
of $\cal A$ on ${\cal H}_\Omega$ can then be defined by: 
$\pi_\Omega(\alpha) |\psi_\beta\rangle = |\psi_{\alpha\beta}\rangle$;
indeed it satisfies: $\pi_\Omega(\alpha)\, \pi_\Omega(\beta)=\pi_\Omega(\alpha\beta)$
and $\pi_\Omega(\alpha^\star)=\big[\pi_\Omega(\alpha)]^\dagger$.
The element $|\psi_{1_{\cal A}}\rangle\equiv|\Omega\rangle$ of ${\cal H}_\Omega$
is cyclic with respect to $\pi_\Omega$, as any element $|\psi_\alpha\rangle$ in
${\cal H}_\Omega$ can be written as $|\psi_\alpha\rangle=\pi_\Omega(\alpha)|\Omega\rangle$,
or in more precise terms, $\pi_\Omega({\cal A})|\Omega\rangle$ is dense in ${\cal H}_\Omega$.
If the state $\Omega$ is not faithful, the same construction holds by identifying
${\cal H}_\Omega$ with the completion of ${\cal A}/{\cal N}_\Omega$, where ${\cal N}_\Omega$
is the kernel of the form $\langle\cdot |\cdot\rangle$ defined above.
\end{proof}

This result makes apparent that the notion of Hilbert space associated to a quantum system
is not a primary concept, but an emergent tool, a consequence of the $C^\star$-algebra structure
of the system observables.
Further, the whole construction sketched above is unique up to unitary transformations.
Indeed, if $\pi'_\Omega$ is another representation of $\cal A$ on a Hilbert space
${\cal H}'_\Omega$ with cyclic vector $|\Omega'\rangle$ such that
$\Omega(\alpha)=\langle\Omega'| \pi'_\Omega(\alpha) |\Omega'\rangle$ for all $\alpha\in{\cal A}$,
then $\pi_\Omega$ and $\pi'_\Omega$ are unitarily equivalent, {\it i.e.} there exists an
isometry $U:{\cal H}_\Omega \to {\cal H}'_\Omega$ such that
$U\, \pi_\Omega U^{-1}=\pi'_\Omega$.

\medskip
\noindent
$\bullet$ {\bf Reducibility and phases}\hfill\break
A representation $\pi$ of the algebra $\cal A$ on an Hilbert space ${\cal H}$ is irreducible if ${\cal H}$
and the null space are the only closed subspaces invariant under the action of $\pi({\cal A})$.
One can prove that the GNS-representation $\pi_\Omega$ is irreducible if and only if the state $\Omega$ is pure.
When the representation $\pi_\Omega$ is not irreducible, it can be decomposed in general into
the direct sum of irreducible representations $\pi^{(r)}_\Omega$:
\begin{equation}
\pi_\Omega=\oplus_r\ \pi^{(r)}_\Omega\ ,
\label{2.3}
\end{equation}
and similarly, also the Hilbert space ${\cal H}_\Omega$ decomposes into the direct sum of
invariant subspaces ${\cal H}^{(r)}_\Omega$ carrying the irreducible representation $\pi^{(r)}_\Omega$:
\begin{equation}
{\cal H}_\Omega=\oplus_r\ {\cal H}^{(r)}_\Omega\ .
\label{2.4}
\end{equation}
As we shall see in the following, irreducibility is an important issue entering the classification
of entangled states.

Any vector $|\psi\rangle\in {\cal H}_\Omega$ defines a new GNS-representation via the state
$\Omega_\psi$ defined by: $\Omega_\psi(\alpha)=\langle\psi| \pi_\Omega(\alpha) |\psi\rangle$,
for all $\alpha\in{\cal A}$. It turns out that the new state $\Omega_\psi$ give rise to
a GNS-representation unitarily equivalent to the one constructed over $\Omega$; in other terms,
${\cal H}_\Omega$ and ${\cal H}_{\Omega_\psi}$ can be identified, as the two representations
$\pi_\Omega$ and $\pi_{\Omega_\psi}$. Similarly, also a density matrix $\rho$ on
${\cal H}_\Omega$ defines a state $\Omega_\rho$ on $\cal A$ trough the identification
$\Omega_\rho(\alpha)={\rm Tr}[\rho\, \pi_\Omega(\alpha)]$, while the corresponding GNS-representation
$\pi_{\Omega_\rho}$ can be expressed in terms of representations equivalent to the representation
$\pi_\Omega$. The set of all states of the form $\Omega_\rho$ forms the so-called {\it folium}
of the representation $\pi_\Omega$; it contains all the states accessible by the
operators $\pi_\Omega(\alpha)$, $\alpha\in {\cal A}$, and constitute 
a {\sl quantum phase} of the physical systems. Systems with infinitely many degrees of freedom,
as in many-body physics, exhibit in general more than one inequivalent phase, {\it i.e.}
they admit more than one inequivalent representation of the associated observable algebra $\cal A$.

\section{Entanglement in Algebraic Quantum Mechanics}

As seen in the previous Section, given any quantum system, once a state $\Omega$ 
on its algebra of observables $\cal A$ is chosen, {\it i.e.} a set of expectation values for the elements 
of $\cal A$ are fixed, one can always construct the associated Hilbert space ${\cal H}_\Omega$ 
and use it for its description. This space contains a reference vector 
$|\Omega\rangle$ through which one can generate the whole ${\cal H}_\Omega$ 
by applying to it elements of $\cal A$.
In more physical terms, all states of the system can be obtained from $|\Omega\rangle$ by the action
of all possible physically acceptable operations.

This algebraic approach to quantum physics turns out to be the most suitable for discussing issues
related to the notions of quantum non-locality and entanglement in very general terms:
it does not make explicit reference to the specific structure of the system under study,
that can indeed be formed even by an infinite number of elementary constituents
and thus possibly exhibiting more than one physical phase.
Although the definition of separability and entanglement within this approach
to quantum theory has been introduced long ago \hbox{\cite{Werner1}-\cite{Summers}}, only recently it has been
applied to characterize non-classical correlations in systems involving identical particles,
both in the case of bosons and fermions \cite{Benatti1}-\cite{Benatti7}.

Given a quantum system, its observable algebra $\cal A$ and a state $\Omega$ on it, one immediately
faces a problem with the standard, textbook definition of separability: 
the associated GNS Hilbert space ${\cal H}_\Omega$ does not result in general a tensor
product of single-particle Hilbert spaces, and therefore the usual notion of
entanglement, based on this structure, is inapplicable.

In line with the characterization of a physical system through its algebra of observables $\cal A$,
one should instead focus the attention on this algebra 
rather than on the Hilbert space ${\cal H}_\Omega$; in this way, the presence of entanglement 
can be identified by the existence of non-classical correlations among mean values
of system observables, belonging to different subalgebras of $\cal A$.
As a preliminary step, it is then necessary to introduce the notion of partition
of the operator algebra $\cal A$.
In the following, we shall consider operator algebras constructed by means of elementary mode operators,
{\it e.g.} annihilation and creation operators, generating an algebra $\cal A$ either of boson or fermion
character; these algebras can be infinite dimensional.
Within this general framework, we then introduce the following basic definition:

\begin{definition}
\label{def1}
An {\bf algebraic bipartition} of the operator algebra ${\cal A}$ is any pair
$({\cal A}_1, {\cal A}_2)$ of subalgebras of ${\cal A}$, namely
${\cal A}_1, {\cal A}_2\subset {\cal A}$, such that 
${\cal A}_1 \cap {\cal A}_2={\bf 1}_{\cal A}$, 
that is they can share only scalar multiples of the identity. Further, in the boson case
the two subalgebras are assumed to commute, 
$\big[A_1\,,\,A_2\big]=0$ for all $A_{i}\in{\cal A}_i$, $i=1,2$,
while in the fermion case, only the even part ${\cal A}^{e}_1$ of ${\cal A}_1$ is required to commute
with the whole ${\cal A}_2$, or ${\cal A}^e_2$ with ${\cal A}_1$; in general,
the even part ${\cal A}^{e}$ of a fermion algebra ${\cal A}$ 
is defined as the norm closure of the algebra of polynomials constructed with even powers
of elementary mode operators.
\end{definition}

\noindent
{\bf Remark:} This definition differs from the one given in \cite{Werner1}-\cite{Moriya1},
in which elements belonging to different bipartions are required to commute, both for bosons
and fermions; as alredy motivated in \cite{Benatti5}-\cite{Benatti7} and further
discussed in the next Section, the previous, more general definition allows
for a more physically complete treatment of fermion entanglement.\hfill$\Box$

\smallskip

\noindent
In general the linear span of products of elements of the two subalgebras 
${\cal A}_1$ and ${\cal A}_2$ need not reproduce the whole
algebra ${\cal A}$, {\it i.e.} 
\hbox{${\cal A}_1 \cup {\cal A}_2\subset {\cal A}$}.
However, in the cases of partitions defined in terms of modes, as discussed below,
one has ${\cal A}_1 \cup {\cal A}_2={\cal A}$, a condition that will be hereafter always assumed.

Any algebraic bipartition encodes in a natural way the definition of the system local observables:
 
\medskip
\noindent
\begin{definition}
An element of ${\cal A}$ is said to be local with respect to
a given bipartition $({\cal A}_1, {\cal A}_2)$, or simply 
$({\cal A}_1, {\cal A}_2)$-{\bf local}, if it is the product $\alpha_1 \alpha_2$ of an element 
$\alpha_1$ of ${\cal A}_1$ and an element $\alpha_2$ in ${\cal A}_2$.
\end{definition}

From this notion of operator locality, a natural definition of state separability and entanglement
follows \cite{Benatti1, Benatti6}:%
\footnote{As already observed, the algebras $\cal A$ need not be finitely generated, so that
the sums appearing below could in principle contain an infinite number of terms;
in such a case, we shall assume their convergence in a proper topology.}
\medskip
\noindent
\begin{definition}
A state $\Omega$ on the algebra ${\cal A}$ will be called {\bf separable} with
respect to the bipartition $({\cal A}_1, {\cal A}_2)$ if the expectation $\Omega(\alpha_1 \alpha_2)$ 
of any local operator $\alpha_1 \alpha_2$ can be decomposed into a linear convex combination of
products of expectations:
\begin{equation}
\Omega(\alpha_1 \alpha_2)=\sum_k\lambda_k\, \Omega_k^{(1)}(\alpha_1)\, \Omega_k^{(2)}(\alpha_2)\ ,\qquad
\lambda_k\geq0\ ,\quad \sum_k\lambda_k=1\ ,
\label{3.1}
\end{equation}
where $\Omega_k^{(1)}$ and $\Omega_k^{(2)}$ are given states on ${\cal A}$;
otherwise the state $\Omega$ is said to be {\bf entangled} with respect the bipartition
$({\cal A}_1, {\cal A}_2)$.
\end{definition}

\noindent
{\bf Remarks:}\hfill\break
{\sl i)} This generalized definition of separability 
can be easily extended to the case
of more than two partitions; 
for instance, in the case of an $n$-partition, Eq.(\ref{3.1})
would extend to:
\begin{equation}
\Omega(\alpha_1 \alpha_2\cdots \alpha_n)=\sum_k\lambda_k\, \Omega_k^{(1)}(\alpha_1)\, 
\Omega_k^{(2)}(\alpha_2)\cdots\Omega_k^{(n)}(\alpha_n)\,\ ,\quad
\lambda_k\geq0\ ,\quad \sum_k\lambda_k=1\ .
\label{3.2}
\end{equation}
\noindent
{\sl ii)} When dealing with systems of $N$ distinguishable constituents, the algebra $\cal A$ usually acts on a Hilbert space ${\cal H}$; if  the state $\Omega$ on $\cal A$ is {\it normal}, {\it i.e.} it can be represented by a
density matrix $\rho_\Omega$, so that
$\Omega(\alpha)={\rm Tr}[\rho_\Omega\, \alpha]$, for any $\alpha\in {\cal A}$,
{\sl Definition 3} gives the standard notion of separability, namely $\rho_\Omega$
can be expressed as a convex combination of product states:
\begin{equation}
\rho_\Omega=\sum_k p_k\, \rho_k^{(1)}\otimes\rho_k^{(2)}\otimes\ldots\otimes\rho_k^{(N)}\ ,
\qquad p_k\geq 0\ ,\quad \sum_k p_k=1\ ,
\label{3.3}
\end{equation}
$\rho^{(i)}$ representing a state for the $i$-th constituent.
Indeed, in this case, the partition of
the system into its elementary constituents induces a natural tensor
product decomposition both of the Hilbert space ${\cal H}$ and of the
algebra $\cal A$ of its operators; a direct application of the condition (\ref{3.2}) to
this natural tensor product multipartition immediately yields the decomposition (\ref{3.3}).

\noindent
{\sl iii)} In systems of identical particles
there is no {\it a priori} given, natural bipartition to be used for the definition
of separability; therefore,
issues about entanglement and non-locality 
are meaningful {\it only} with reference to a choice of a specific partition in the 
associated operator algebra
\cite{Zanardi}-\cite{Benatti7}; 
this general observation, often overlooked, is at the origin of
much confusion in the recent literature.\hfill$\Box$

\smallskip

\noindent
When the state $\Omega$ is pure, the separability condition (\ref{3.1}) simplifies,
and the following result holds:

\begin{lemma}
Pure states $\Omega$ on the operator algebra $\cal A$ are separable with respect to a given
bipartition $({\cal A}_1,{\cal A}_2)$ if and only if
\begin{equation}
\Omega(\alpha_1 \alpha_2)=\Omega(\alpha_1)\, \Omega(\alpha_2)\ ,
\label{3.4}
\end{equation}
for all local operators $\alpha_1 \alpha_2$.
\end{lemma}
\noindent
In other terms, separable, pure states are just product states. We shall first give a proof
of this result for boson operator algebras, leaving to the next Section
the analysis of fermion algebras.

\begin{proof}

For the {\it if} part of the proof, observe that, according to {\sl Definition 3},
states satisfying (\ref{3.4}) are automatically $({\cal A}_1, {\cal A}_2)$-separable
since they obey (\ref{3.1}) with only one contribution to the convex sum.
For the {\it only if} part of the proof, recall that 
$({\cal A}_1, {\cal A}_2)$-local operators generate the whole boson algebra $\cal A$,
since ${\cal A}_1 \cup {\cal A}_2={\cal A}$; 
then, any element $\alpha\in {\cal A}$ can be written 
as $\alpha=\sum_{i} \alpha^{(1)}_i\,\alpha^{(2)}_i$, with $\alpha^{(1)}_i\in {\cal A}^{(1)}$ 
and $\alpha^{(2)}_i\in {\cal A}^{(2)}$. Therefore, if by hypothesis a state $\Omega$ is separable,
{\it i.e.} it can be written as in (\ref{3.4}) on all $({\cal A}_1, {\cal A}_2)$-local operators, then
\begin{equation}
\Omega(\alpha)=\sum_i \Omega(\alpha^{(1)}_i\,\alpha^{(2)}_i)
=\sum_{ik} \lambda_k\, \Omega_k^{(1)}(\alpha^{(1)}_i)\, \Omega_k^{(2)}(\alpha^{(2)}_i)
=\sum_k \lambda_k\ \omega_k (\alpha)\ ,
\label{3.4-1}
\end{equation}
where $\omega_k$ are linear maps defined on the whole algebra $\cal A$ by the relation
\begin{equation}
\omega_k (\alpha)=\sum_i \Omega_k^{(1)}(\alpha^{(1)}_i)\, \Omega_k^{(2)}(\alpha^{(2)}_i)\ .
\label{3.4-2}
\end{equation}
One easily sees that these maps are positive. Indeed, for any $\alpha\in {\cal A}$ one has
($T$ signifies matrix transposition):
\begin{equation}
\omega_k (\alpha \alpha^\star)={\rm Tr}\Big[ M^{(1)}_k\, \big(M^{(2)}_k)^T \Big]\ ,\quad
\big[M^{(\ell)}_k\big]_{ij}=\Omega_k^{(\ell)}\Big(\alpha^{(\ell)}_i\, \big(\alpha^{(\ell)}_j\big)^\star\Big)\ ,
\quad \ell=1,2\ ,
\label{3.4-3}
\end{equation}
with the matrices $M^{(1)}_k$, $M^{(2)}_k$ hermitian and positive; since the trace of the product
of two positive matrices is positive, one immediately gets: $\omega_k (\alpha \alpha^\star)\geq 0$. In addition,
the maps $\omega_k$ are normalized, 
$\omega_k({\bf 1})=\Omega_k^{(1)}({\bf 1})\, \Omega_k^{(2)}({\bf 1})=1$, and therefore represent states
for the algebra $\cal A$. But since $\Omega$ is pure by hypothesis, only one term in the convex
combination (\ref{3.4-1}) must be different from zero. Dropping the superfluous label $k$,
we have then found: $\Omega(\alpha^{(1)}\, \alpha^{(2)})=
\Omega^{(1)}(\alpha^{(1)})\, \Omega^{(2)}(\alpha^{(2)})$. By separately taking $\alpha^{(1)}$
and $\alpha^{(2)}$ to coincide with the identity operator, one finally obtains the result (\ref{3.4}).
\end{proof}

\noindent
{\bf Remark:} As we shall see explicitly later on, given a bipartition of the algebra $\cal A$,
a pure separable state $\Omega$ on it
is in general no longer pure when restricted to a proper subalgebra ${\cal B}\subset {\cal A}$;
nevertheless, since in any case it obeys the condition (\ref{3.4}), it will remain separable.\hfill$\Box$

\smallskip

Using the previous {\sl Definitions} and the result of {\sl Lemma 1}, one can 
study the entanglement with respect to a given bipartition $({\cal A}_1,{\cal A}_2)$ of 
the boson algebra $\cal A$ of states in a {\it folium} (see the discussion at the end of Section 2)
of the representation $\pi_\Omega$ corresponding to a given state $\Omega$ on $\cal A$, assuming
$\Omega$ to be separable.

The specific separability condition (\ref{3.4}) allows obtaining the generic form of any pure separable state:
\begin{proposition}
Let $({\cal A},\, \Omega)$ be operator algebra and state associated to a given boson quantum system and
assume $\Omega$ to be separable with respect to a given bipartition $({\cal A}_1,\, {\cal A}_2)$.
Then a normalized pure state $|\psi\rangle$ in the GNS-Hilbert space ${\cal H}_\Omega$ is 
$({\cal A}_1,\, {\cal A}_2)$-separable if and only if it can be written in the form
\begin{equation}
|\psi\rangle=\pi_\Omega(\beta^{(1)})\, \pi_\Omega(\beta^{(2)})\ |\Omega\rangle \ ,
\label{3.5}
\end{equation}
with $\beta^{(i)}\in{\cal A}_i$, $i=1,\ 2$, while $\pi_\Omega(\beta^{(i)})$ denote the corresponding 
operator representation on the Hilbert space ${\cal H}_\Omega$.
\end{proposition}

\begin{proof}
For the {\it if} part of the proof, notice that the normalization condition
together with the assumed $({\cal A}_1,\, {\cal A}_2)$-separability of $\Omega$ yield:
\begin{eqnarray*}
\nonumber
&&\hskip-.5cm
\langle\psi | \psi\rangle=\langle\Omega\vert\pi^\dag_\Omega(\beta^{(2)})\pi^\dag_\Omega(\beta^{(1)})\pi_\Omega(\beta^{(1)})\,\pi_\Omega(\beta^{(2)})\vert\Omega\rangle\\
\label{aid}
&&\hskip 5cm
=\langle\Omega\vert\pi^\dag_\Omega(\beta^{(1)})\,\pi_\Omega(\beta^{(1)})\vert\Omega\rangle\,\langle\Omega\vert\pi^\dag_\Omega(\beta^{(2)})\,\pi_\Omega(\beta^{(2)}\vert\Omega\rangle=1\ .
\end{eqnarray*}
Using this result and again the separability of $\Omega$, one then has
\begin{eqnarray*}
&&\hskip-.5cm
\langle\psi\vert\pi_\Omega(\alpha^{(1)})\pi_\Omega(\alpha^{(2)})\vert\psi\rangle=
\langle\Omega\vert\pi^\dag_\Omega(\beta^{(2)})\,\pi_\Omega(\alpha^{(2)})\,\pi_\Omega(\beta^{(2)})\vert\Omega\rangle\,
\langle\Omega\vert\,\pi^\dag_\Omega(\beta^{(1)})\pi_\Omega(\alpha^{(1)})\,\pi_\Omega(\beta^{(1)})\vert\Omega\rangle\\
&&\hskip 9cm
=\langle\psi\vert\pi_\Omega(\alpha^{(2)})\vert\psi\rangle\,\langle\psi\vert\pi_\Omega(\alpha^{(1)})\vert\psi\rangle\ .
\end{eqnarray*}

For the {\it only if} part of the proof, 
observe that due to the cyclicity of the GNS state $|\Omega\rangle$, one can surely write
$|\psi\rangle=\pi_\Omega(\beta)\, |\Omega\rangle$, for some $\beta\in{\cal A}$.
Further, since ${\cal A}_1 \cup {\cal A}_2 = {\cal A}$, $\beta$ can be written
as combination of suitable local operators, $\beta=\sum_i \beta_i^{(1)}\, \beta_i^{(2)}$,
with $\beta_i^{(1)}\in{\cal A}_1$ and $\beta_i^{(2)}\in{\cal A}_2$.
Then, for any local operator $\alpha^{(1)}\alpha^{(2)}$, the separability condition (\ref{3.1}) implies
\begin{eqnarray}
\nonumber
\langle\psi| \pi_\Omega(\alpha^{(1)})\, \pi_\Omega(\alpha^{(2)} )|\psi\rangle
&=& \sum_{i,j}\langle\psi^{(2)}_i\vert\psi^{(2)}_j\rangle\,
\langle \psi^{(1)}_i| \pi_\Omega(\alpha^{(1)}) |\psi^{(1)}_j\rangle\\
&&\hskip 2cm\times\sum_{r,s}\langle\psi^{(1)}_{r}\vert\psi^{(1)}_{s}\rangle\, 
\langle \psi^{(2)}_{r}| \pi_\Omega(\alpha^{(2)}) |\psi^{(2)}_{s}\rangle\ ,
\label{3.6}
\end{eqnarray}
where 
\begin{equation}
|\psi_i^{(\ell)}\rangle= \pi_\Omega(\beta_i^{(\ell)})\, |\Omega\rangle,\quad \ell=1,2\ . 
\label{3.6-0}
\end{equation}
The matrices $\langle\psi^{(\ell)}_i\vert\psi^{(\ell)}_j\rangle$ are hermitian and positive semi-definite and therefore
they can be diagonalized by suitable unitary transformations:
\begin{equation}
\langle\psi^{(\ell)}_i\vert\psi^{(\ell)}_j\rangle=\sum_r \big[U^{(\ell)^\dagger}\big]_{ir}\, \lambda_r^{(\ell)}\,
\big[U^{(\ell)}\big]_{rj}\ ,\quad U^{(\ell)}\, U^{(\ell)^\dagger}=1\ ,\quad \lambda_r^{(\ell)}\geq0\ ,\quad
\ell=1,2\ .
\label{3.6-1}
\end{equation}
Then, by defining
\begin{equation}
|\phi^{(1)}_r\rangle= \sum_j \big[U^{(2)}\big]_{rj}\, |\psi_j^{(1)}\rangle\ , \qquad
|\phi^{(2)}_r\rangle= \sum_j \big[U^{(1)}\big]_{rj}\, |\psi_j^{(2)}\rangle\ ,
\nonumber
\end{equation}
and recalling again that any element $\alpha\in{\cal A}$ can be decomposed in terms of local operators,
$\alpha=\sum_i \alpha_i^{(1)}\, \alpha_i^{(2)}$, using (\ref{3.6}) one can write:
\begin{equation}
\langle\psi| \pi_\Omega(\alpha)|\psi\rangle=
\sum_{r,s} \lambda_r^{(1)} \lambda_s^{(2)}\ \omega_{rs}(\alpha)\ ,
\label{3.6-2}
\end{equation}
with
$$
\omega_{rs}(\alpha)=\sum_i \langle \phi^{(1)}_r | \alpha^{(1)}_i | \phi^{(1)}_r\rangle\,
\langle \phi^{(2)}_s | \alpha^{(2)}_i | \phi^{(2)}_s\rangle\ .
$$
As in the proof of {\sl Lemma 1}, one easily sees that the linear maps $\omega_{rs}$ are actually
(un-normalized) states on $\cal A$. Then, by the purity of the state $|\psi\rangle$,
the convex combination on the r.h.s. of (\ref{3.6-2}) must contain just one term.
This implies that there are just two positive constants $\lambda^{(1)}$ and $\lambda^{(2)}$ and
therefore no sum over $r$ in (\ref{3.6-1}): 
$\langle\psi^{(\ell)}_i\vert\psi^{(\ell)}_j\rangle=[V_i^{(\ell)}]^*\, V_j^{(\ell)}$, $\ell=1,2$,
with $V_i^{(\ell)}=\sqrt{\lambda^{(\ell)}}\, U_i^{(\ell)}$, $U_i^{(\ell)}\in {\bf C}$.
These inner products
are thus in factorized form and this is possible only if the vectors 
$|\psi_i^{(\ell)}\rangle$ are all proportional: 
$|\psi_i^{(\ell)}\rangle\equiv V_i^{(\ell)}\, |\psi^{(\ell)}\rangle$. Recalling their definition given in
(\ref{3.6-0}), this in turn implies: $\beta_i^{(\ell)}=V_i^{(\ell)}\ \beta^{(\ell)}$, up to operators
annihilating $|\Omega\rangle$; as a consequence 
the factorized form (\ref{3.5}) follows.
\end{proof}

As an immediate consequence of this result, one has the following {\sl Corollary},
that it will be very useful in characterizing the general form
of entangled states in ${\cal H}_\Omega$.

\begin{corollary}
Let $({\cal A},\, \Omega\, ,\pi_\Omega)$ be the triple characterizing the operator algebra, 
state and GNS-representation of a given boson quantum system; given any bipartition $({\cal A}_1, {\cal A}_2)$
of $\cal A$, with $\Omega$ separable, it is always possible to choose a basis in the corresponding Hilbert space
${\cal H}_\Omega$ made of separable pure states.
\end{corollary}

\begin{proof}
In each of the two subalgebras ${\cal A}_1$, ${\cal A}_2$ viewed as linear spaces fix
a basis $\{e_k^{(1)} \}$, $\{e_\ell^{(2)} \}$, so that any element $\alpha^{(i)}\in{\cal A}_i$,
$i=1,2$, can be decomposed as linear combination of the basis elements,
$\alpha^{(i)}=\sum_k c_k^{(i)}\, e_k^{(i)}$, with $c_k^{(i)}$ complex coefficients.
Then,
the following set of vectors 
$|k;\ell\rangle = \pi_\Omega( e_k^{(1)})\, \pi_\Omega( e_\ell^{(2)})\, |\Omega\rangle$, 
obtained by applying products of the chosen basis elements to the cyclic vector $|\Omega\rangle$, 
are pure and separable by construction; further, they span the whole ${\cal H}_\Omega$, since otherwise 
the condition ${\cal A}_1 \cup {\cal A}_2= {\cal A}$ would not be satisfied.
\end{proof}

In many relevant cases, one can choose the basis elements $e_k^{(i)}\in{\cal A}_i$, $i=1,2$,
so that the resulting basis $\{ |k;\ell\rangle \}$ of ${\cal H}_\Omega$ results orthonormal.
In such cases, any (normal) mixed state $\Omega_\rho$, represented by the density matrix
$\rho$ on ${\cal H}_\Omega$, can be decomposed as
\begin{equation}
\rho=\sum_{j,k,\ell,m} \rho_{jk,\ell m}\ |j;k\rangle\langle \ell;m |\ ,
\qquad \sum_{j,k} \rho_{jk,jk}= 1\ .
\label{3.8}
\end{equation}
A density matrix $\rho_D$ in diagonal form,
\begin{equation}
\rho_D=\sum_{j,k} \rho_{jk,jk}\ |j;k\rangle\langle j;k |\ ,
\label{3.9}
\end{equation}
is clearly separable, being a convex combination of projectors on separable pure states.
More explicitly, for any local operator $\alpha^{(1)}\alpha^{(2)}$, one has:
\begin{eqnarray}
\nonumber
&&\hskip - 1.5cm\Omega_{\rho_D}(\alpha^{(1)}\alpha^{(2)})=
{\rm Tr}\Big[\rho_D\, \pi_\Omega(\alpha^{(1)})\, \pi_\Omega(\alpha^{(2)})\Big]\\
&&\hskip - 1cm=\sum_{j,k} \rho_{jk,jk}\ \langle\Omega | \big[\pi_\Omega(e_j^{(1)})\big]^\dagger 
\pi_\Omega(\alpha^{(1)})\, \pi_\Omega(e_j^{(1)}) |\Omega\rangle\
\langle\Omega | \big[\pi_\Omega(e_k^{(2)})\big]^\dagger 
\pi_\Omega(\alpha^{(2)})\, \pi_\Omega(e_j^{(2)}) |\Omega\rangle\ ,
\label{3.10}
\end{eqnarray}
which is precisely of the separable form (\ref{3.1}). 
This observation, together with {\sl Proposition 1}
and the fact that generic separable mixed states belong to the convex hull of pure separable states, 
can be used to characterize separable mixed states:

\begin{corollary}
A mixed boson state $\rho$ as in (\ref{3.8}) is
separable with respect to the given bipartition $({\cal A}_1, {\cal A}_2)$ if and only if it is a
convex combination of projectors onto pure $({\cal A}_1, {\cal A}_2)$-separable
states; otherwise, the state $\rho$ is $({\cal A}_1, {\cal A}_2)$-entangled.
\end{corollary}

In general, to determine whether a given density matrix $\rho$ can be written in diagonal, separable 
form is a hard task and one is forced to rely on suitable separability tests, that however
are in general not exhaustive. As discussed in the next Section, 
one of such tests is peculiar to fermion systems, and it is connected to the anticommuting character
of the corresponding operator algebra $\cal A$.

\section{Entanglement in fermion systems}

In this Section, we extend the results previously obtained in the case of boson systems
to many-body systems made of fermion elementary constituents.%
\footnote{Part of the results of this Section have being already
discussed in \cite{Benatti5}-\cite{Benatti7}.}
Adopting a second-quantized point of view, for fermion systems the observable algebra $\cal A$
coincides with the complex algebra ${\cal A}_f$ of canonical anticommutation relations.
It is generated by elements $a_i$, $a_i^\star$ obeying the relations:
\begin{equation}
\{a_i,\,a^\star_j\}\equiv a_i\,a^\star_j+a^\star_j\, a_i=\delta_{ij}\ ,\quad\{a_i,\,a_j\}=
\{a_i^\star,\,a^\star_j\}=\, 0\ ,\qquad i,\, j=1,2,\ldots,M\ ,
\label{4.1}
\end{equation}
where for simplicity we have assumed that the fermions can occupy $M$ different modes
(with $M$ possibly infinite). This framework is quite general and can accommodate various
situations arising in atomic and condensed matter physics; 
in particular, it can be used to describe ultracold fermions confined in
optical lattices \cite{Leggett1}-\cite{Yukalov}. 
The norm closure of all polynomials in the creation and annihilation
operators gives the full fermion operator algebra~${\cal A}_f$.

This algebra has a natural gradation in terms of its even and odd part:

\begin{definition}
Introduce the automorphism $\vartheta$ of ${\cal A}_f$ defined by its action on
the basic operators $a_i$ and $a_i^\star$ as follows: $\vartheta(a_i)=-a_i$, 
$\vartheta(a_i^\star)=-a_i^\star$. Then, the even component ${\cal A}_f^e$
of ${\cal A}_f$ is the subset containing the elements $\alpha^e\in {\cal A}_f$
such that $\vartheta(\alpha^e)=\alpha^e$, while the odd component ${\cal A}_f^o$
of ${\cal A}_f$ consists of elements $\alpha^o\in {\cal A}_f$,
for which $\vartheta(\alpha^o)=-\alpha^o$.
\end{definition}

\noindent
The even part ${\cal A}_f^e$ is a subalgebra of ${\cal A}_f$, the one generated
by even polynomials in all creation and annihilation operators; on the other hand,
${\cal A}_f^o$ is only a linear space and not an algebra, since the product
of two odd elements is even. Nevertheless, using the two projectors
${\cal P}^e=(1+\vartheta)/2$ and ${\cal P}^o=(1-\vartheta)/2$, any element $\alpha\in {\cal A}_f$
can be decomposed in its even $\alpha^e\equiv {\cal P}^e(\alpha)$ and odd 
$\alpha^o\equiv {\cal P}^o(\alpha)$ parts: $\alpha=\alpha^e + \alpha^o$.

A bipartition of the $M$-mode fermion algebra ${\cal A}_f$ can be easily 
obtained by splitting the collection of operators $\{a_i,\, a_i^\star \}$
into two disjoint sets $\{a_i,\, a_i^\star |\, i=1,2,\ldots,m\}$
and $\{a_j,\, a_j^\star |\, j=m+1,m+2,\ldots,M\}$. All polynomials in the first set
(together with their norm closures) form a subalgebra ${\cal A}_1$,
while the second set generates a subalgebra ${\cal A}_2$. These two algebras
have only the unit element in common and further ${\cal A}_1 \cup {\cal A}_2 = {\cal A}_f$.
Further, one defines
the even ${\cal A}_i^e$ and odd ${\cal A}_i^o$ components of the two
subalgebras ${\cal A}_i$, $i=1,2$, as done above for the full algebra ${\cal A}_f$,
through the automorphism~$\vartheta$. Only the operators of the first partition 
belonging to the even component $\mathcal{A}_1^e$
commute with any operator of the second partition and, similarly, only the even operators
of the second partition commute with the whole subalgebra $\mathcal{A}_1$.
Recalling now {\sl Definition~1}, $({\cal A}_1, {\cal A}_2)$ 
is indeed an algebraic bipartition of ${\cal A}_f$; in practice, it is determined by
the choice of the integer $m$, with $0<m< M$. 

Let us now come back to the definition of separability
introduced in {\sl Definitions~1-3} and to the apparent difference 
with which it treats bosonic and fermionic systems.
As already noticed, in the boson case, the two subalgebras 
${\cal A}_1$, ${\cal A}_2$ defining the algebraic bipartition $({\cal A}_1,\, {\cal A}_2)$
naturally commute, {\it i.e.} that each element $\alpha_1$ of the operator algebra ${\cal A}_1$ 
commutes with any element $\alpha_2$ in ${\cal A}_2$.
Instead, in the case of fermion systems 
the two subalgebras ${\cal A}_1$, ${\cal A}_2$ do not in general commute.
Nevertheless, in such systems only selfadjoint operators 
belonging to the even components ${\cal A}_1^e$, ${\cal A}_2^e$ 
qualify as physical observables and these do commute as required
by the definition of bipartition.

At this point, two different attitudes are possible regarding the definition of separability
expressed by the condition (\ref{3.1}): {\it i)} allow in it all operators from the two subalgebras
${\cal A}_1$, ${\cal A}_2$, or {\it ii)} restrict all
considerations to observables only. The first approach is in line with the notion of
``microcausality'' adopted in {\sl constructive} quantum field theory \cite{Streater,Strocchi4},
where the emphasis is on quantum fields, which are required either to commute (boson fields)
or anticommute (fermion fields) if defined on (causally) disjoint regions. 
On the other hand, the second point of view reminds
of the notion of ``local commutativity'' in {\sl algebraic} quantum field theory \cite{Emch,Haag},
where only observables are considered, assumed to commute if localized in disjoint regions.

These two points of view are not equivalent, as it can be appreciated by the following
simple example. Let us consider the system consisting of just one fermion that can occupy
two modes, $M=2$, with the bipartition defined by the two modes themselves. 
In the standard Fock representation, {\it i.e.} 
the GNS-construction built out of the vacuum state $|\Omega_0\rangle$,
such that $\pi_{\Omega_0}(a_i)|\Omega_0\rangle=\,0$, $i=1,2$ (see below for details),
consider the following state:
\begin{equation}
\Omega=|\phi\rangle\langle \phi|\ ,\qquad 
|\phi\rangle=\frac{1}{\sqrt2}\Big(|1,0\rangle + |0,1\rangle\Big)\ ,
\label{4.1.1}
\end{equation}
combination of the two manifestly separable Fock states 
$|1,0\rangle=\pi_{\Omega_0}(a_1^\star)|\Omega_0\rangle$
and
$|0,1\rangle=\pi_{\Omega_0}(a_2^\star)|\Omega_0\rangle$.
Clearly, it appears to be entangled and indeed, as discussed in \cite{Benatti6},
in a suitable $N$-fermion generalization, its quantum non-locality can be used
in quantum metrology to achieve sub shot-noise accuracy in parameter estimation.

Nevertheless,
in the second approach mentioned above, it is found to satisfy the condition (\ref{3.1}), hence to be separable. 
Indeed, only observables, {\it i.e.} selfadjoint, even operators,
can be used in this case as $\alpha_1$ and $\alpha_2$; in practice, only the two partial number operators
$a_1^\star a_1$ and $a_2^\star a_2$ 
together with the identity are admissible, and for these observables
the state (\ref{4.1.1}) behaves as the separable state $(|0,1\rangle\langle 0,1| + |1,0\rangle\langle 1,0|)/2$.
Different is the situation within the first approach: in this case, all operators
are admissible, for instance $a_1^\star$ and $a_2$,
which indeed prevent the separability condition (\ref{3.1}) to be satisfied.
In view of this, as in \cite{Benatti5}-\cite{Benatti7}, we here advocate and adopt 
the first point of view, {\it i.e.} point {\it i)} above: 
it gives a more general and physically
complete treatment of fermion entanglement.

In this respect, it should be added that
the anticommuting character
of the fermion algebra gives stringent constraints on the form of the states
defined on it, specifically on the ones that can be represented as a product
of other states.

As for any operator algebra, a state on ${\cal A}_f$ is a positive, linear functional
$\Omega: {\cal A}_f \to \mathbb{C}$. Then the following result holds
(see \cite{Araki,Benatti6} for the rather simple proof):
\begin{lemma}
Consider a bipartition $({\cal A}_1, {\cal A}_2)$ of the fermion algebra ${\cal A}_f$ 
and two states $\Omega_1$, $\Omega_2$ on it. Then, the linear 
functional $\Omega$ on ${\cal A}_f$ defined by 
$\Omega(\alpha_1\alpha_2)=\Omega_1(\alpha_1)\,\Omega_2(\alpha_2)$ 
for all $\alpha_1\in{\cal A}_1$ and $\alpha_2\in{\cal A}_2$ is a 
state on ${\cal A}_f$ only if at least one $\Omega_i$ vanishes on the odd component of ${\cal A}_i$.
\end{lemma}
\noindent
This result implies that the product 
$\Omega_k^{(1)}(\alpha_1)\, \Omega_k^{(2)}(\alpha_2)$ in the r.h.s. of (\ref{3.1})
in {\sl Definition~3}, vanishes whenever 
$\alpha_1$ and $\alpha_2$ are both odd. Since the even component ${\cal A}_1^e$ commutes with
the entire subalgebra ${\cal A}_2$, and similarly ${\cal A}_2^e$ commutes with ${\cal A}_1$,
it follows that also for fermions the decomposition~(\ref{3.1}) is non-trivial only for local operators
$\alpha_1\alpha_2$ such that $[\alpha_1,\, \alpha_2]=\,0$, thus making the definition of separability
it encodes completely analogous to the one for bosons.

\medskip
\noindent
{\bf Remark:} Given a partition $({\cal A}_1, {\cal A}_2)$ of ${\cal A}_f$, 
consider a product state $\Omega$ such that $\Omega(\alpha_1\alpha_2)=\Omega(\alpha_1)\,\Omega(\alpha_2)$,
for all $\alpha_1\in{\cal A}_1$ and $\alpha_2\in{\cal A}_2$; because of {\sl Lemma 2}, it must be zero on
the odd elements of at least one partition. Indeed, fixing two odd elements $\alpha_1^o$ and $\alpha_2^o$
in the two partition, by {\sl Lemma 2} at least one of the two expectations $\Omega(\alpha_1^o)$
or $\Omega(\alpha_2^o)$ must be zero. Assume $\Omega(\alpha_2^o)\neq0$; then, again by the previous {\sl Lemma},
one has, 
$\Omega(\beta^o_1\alpha^o_2)=\Omega(\beta^o_1)\,\Omega(\alpha^o_2)=\,0$, and thus $\Omega(\beta^o_1)=\,0$,
for any odd element $\beta_1^o\in {\cal A}_1$.\hfill$\Box$
\medskip

Motivated by this last {\sl Remark}, in the following we shall limit our considerations to states on ${\cal A}_f$ 
that are left invariant by the action of the automorphism $\vartheta$, $\Omega \circ \vartheta = \Omega$,
namely states that are vanishing on the odd component of the fermion algebra: 
this physical, ``gauge-invariance'' condition
is always tacitly assumed in the discussion of any fermion many-body system.%
\footnote{This is apparent in the standard Fock representation of the fermion algebra discussed below, since
the vacuum expectation of an odd number of elements $a_i$ and $a_i^\star$ is always vanishing.}

Within the framework introduced above, most of the results discussed in the previous Section
in the case of boson algebras remain true also for the fermion algebra ${\cal A}_f$.
In particular, the characterization of pure, separable states given
by {\sl Lemma 1} and {\sl Proposition~1} is unaltered; however,
the proofs of these results need refinements in order to comply with
the anticommuting character of ${\cal A}_f$.

\noindent \ - {\it Proof of} {\sl Lemma 1}: The {\it if} part of the proof is unaltered, while for the {\it only if}
part, one notices that also in this case any element $\alpha\in {\cal A}_f$ can be written
as $\alpha=\sum_{i} \alpha^{(1)}_i\,\alpha^{(2)}_i$, with $\alpha^{(1)}_i\in {\cal A}^{(1)}$, 
$\alpha^{(2)}_i\in {\cal A}^{(2)}$; however, as previously observed,
the elements $\alpha^{(\ell)}_i$, $\ell=1,2$, can be decomposed as the sum of their even and odd parts,
so that in the previous decomposition of $\alpha$ in terms of local operators,
one can assume all $\alpha^{(\ell)}_i$ to be of given parity.
The proof than proceeds as in the boson case, since the expressions in (\ref{3.4-1}) and (\ref{3.4-2}) 
are unaltered. The only troubled point is to show that the linear maps 
$\omega_k(\alpha)$ in (\ref{3.4-1}) are really states
on ${\cal A}_f$, {\it i.e.} that $\omega_k(\alpha\,\alpha^\star)\geq0$. This is done by explicit calculation
showing that $\omega_k(\alpha\,\alpha^\star)$ can be expressed as in (\ref{3.4-3}) plus  additional
pieces that are however vanishing due to the result of {\sl Lemma 2}.

\smallskip
\noindent \ - {\it Proof of} {\sl Proposition 1}: In this case it is the {\it if} part of the proof that requires some care,
while the {\it only if} part of the proof proceeds exactly as in the boson case, recalling that
averages of odd operators on the GNS-state $|\Omega\rangle$ vanish.
Given the bipartition $({\cal A}_1, {\cal A}_2)$, let us assume that the pure, 
normalized state $|\psi\rangle$ can be written as in (\ref{3.5}), {\it i.e.}
$|\psi\rangle=\pi_\Omega(\beta^{(1)})\, \pi_\Omega(\beta^{(2)})\ |\Omega\rangle$,
$\beta^{(\ell)}\in {\cal A}_\ell$, $\ell=1,2$;
we have to prove that:
\begin{equation}
\langle\psi\vert\pi_\Omega(\alpha^{(1)})\pi_\Omega(\alpha^{(2)})\vert\psi\rangle
=\langle\psi\vert\pi_\Omega(\alpha^{(2)})\vert\psi\rangle\,\langle\psi\vert\pi_\Omega(\alpha^{(1)})\vert\psi\rangle\ ,
\label{3.1-1}
\end{equation}
for any $\alpha^{(\ell)}\in {\cal A}_\ell$. First of all, notice that if (\ref{3.1-1}) is true for
$\alpha^{(1)}$, $\alpha^{(2)}$ of definite parity, than it is true also for generic $\alpha^{(1)}$,
$\alpha^{(2)}$, since they can be always decomposed as the sum of their even and odd parts.
The proof then splits in four parts, according to all the possible combinations of parities that the elements
$\alpha^{(1)}$, $\alpha^{(2)}$ can take. By writing also $\beta^{(1)}$ and $\beta^{(2)}$
as the sum of their even and odd parts and using the normalization condition $\langle\psi |\psi\rangle=1$,
explicit computation then shows that the result (\ref{3.1-1}) is indeed true, keeping in mind that $\Omega$ is separable
and vanishing on odd elements of ${\cal A}_f$.
\hfill$\Box$

\medskip
As a further consequence of {\sl Lemma 2}, the following criterion of entanglement holds:

\begin{corollary}
Given the bipartition $({\cal A}_1, {\cal A}_2)$ of the fermion algebra ${\cal A}_f$,
if a state $\Omega$ is non vanishing on a local operator $\alpha_1^o\alpha_2^o$, with the two components 
$\alpha_1^o \in {\cal A}_1^o$, $\alpha_2^o \in {\cal A}_2^o$ both belonging to the odd part 
of the two subalgebras, then $\Omega$ is entangled.
\end{corollary}

\begin{proof}
Indeed, if $\Omega(\alpha_1^o \alpha_2^o)\neq0$, then, by {\sl Lemma 2}, $\Omega$ can not be written as
in (\ref{3.1}), and therefore it is entangled.
\end{proof}

\bigskip

The standard GNS-construction for the algebra ${\cal A}_f$ is based on the vacuum state
$\Omega_0$ giving rise to the so-called Fock representation. It is characterized by the
condition \hbox{$\Omega_0(a_i)=0$}, for all annihilation operators $a_i$, or equivalently,
$\pi_{\Omega_0}(a_i)\, |\Omega_0\rangle=\,0$; the corresponding Hilbert space ${\cal H}_{\Omega_0}$
is spanned by the states obtained applying creation operators, 
$\pi_{\Omega_0}(a_i^\star)\equiv\big[\pi_{\Omega_0}(a_i)\big]^\dagger$ to the cyclic vector
$|\Omega_0\rangle$. A basis in ${\cal H}_{\Omega_0}$ is then given by the
set of Fock states:
\begin{equation}
|n_1, n_2,\ldots,n_M\rangle= 
\big[\pi_{\Omega_0}(a_1^\star)\big]^{n_1}\, \big[\pi_{\Omega_0}(a_2^\star)\big]^{n_2}\, \cdots\, 
\big[\pi_{\Omega_0}(a_M^\star)\big]^{n_M}\,|\Omega_0\rangle\ ,
\label{4.2}
\end{equation}
the integers $n_1, n_2, \ldots, n_M$ representing the occupation numbers of the different modes;
due to algebraic relations (\ref{4.1}), they can take only the two values 0 or 1.  
In this representation, the total number $\hat N=\sum_i \hat N_i$, with 
$\hat N_i=\pi_{\Omega_0}(a_i^\star)\, \pi_{\Omega_0}(a_i)$ counting the occupation number
of the $i$-th mode, is a well defined operator on ${\cal H}_{\Omega_0}$; 
as a consequence, Fock states with different occupation numbers result orthogonal.

\noindent
{\bf Remark:} 
Notice that the operator $\hat N$ commutes with all physical observables, 
since coherent mixtures of states with different total occupation number
are not physical, due to the conservation of the fermion parity operator:
we are in presence of a so-called 
superselection rule \cite{Bartlett}-\cite{Moriya2}.\hfill$\Box$
\medskip

One easily sees that the Fock representation of ${\cal A}_f$ is irreducible, so that
the Fock states in (\ref{4.2}) are pure on ${\cal A}_f$. 
Further, they are separable with respect to any fermion bipartition $({\cal A}_1,{\cal A}_2)$ 
as in {\sl Definition 1}, since they are in the product form (\ref{3.4}).
Therefore, they can be used to give a convenient decomposition of any fermion state
in the {\sl folium} of $\Omega_0$, in particular, for any density matrix $\rho$ on ${\cal H}_{\Omega_0}$.

First, note that due to the above mentioned fermion number superselection rule,
a general fermion density matrix can be written as an incoherent superposition
of states $\rho_N$ with a fixed number $N$ of fermions
\begin{equation}
\rho=\sum_N \lambda_N\, \rho_N \ ,\qquad \lambda_N\geq 0\ ,\qquad \sum_N \lambda_N=1\ .
\label{4.3}
\end{equation}
One can then limit the discussion to the states $\rho_N$, with $N$ fixed.
Indeed, notice that two density matrices $\rho_{N_1}$ and $\rho_{N_2}$,
with $N_1\neq N_2$, have supports on orthogonal subspaces of ${\cal H}_{\Omega_0}$;
as a result, the Fock Hilbert space decomposes as ${\cal H}_{\Omega_0}=\oplus_N {\cal H}_N$,
where ${\cal H}_N$ are Hilbert spaces spanned by Fock vectors (\ref{4.2}) 
having exactly $N$ fermions, {\it i.e.} $\sum_i n_i = N$.

As discussed above, a bipartition of ${\cal A}_f$ is given by a partition
of the fermion modes into two disjoint sets, one containing the first $m$ modes,
while the second the remaining $M-m$ ones; we can refer to such a choice 
as the $(m,\ M-m)$-partition. Given such a bipartition
the Fock basis in ${\cal H}_{\Omega_0}$ can be relabeled in a more convenient way as
$| k, \sigma; N-k, \sigma'\rangle$, where 
the integer $k$ gives the number of occupied modes in the first partition,
while $\sigma$ counts the different ways in which these modes 
can be taken out of the available $m$ ($k\leq m$);
similarly, $\sigma'$ distinguishes the ways in which the remaining $N-k$ occupied modes can
be distributed in the second partition.

Then, a generic density matrix $\rho_N$ on ${\cal H}_N$ can be decomposed as
\begin{equation}
\rho_N=\sum_{k,l=N_-}^{ N_+ }\ \sum_{\sigma,\sigma',\tau,\tau'}\
\rho_{k \sigma\sigma', l\tau\tau'}\ | k, \sigma; N-k, \sigma'\rangle \langle l, \tau; N-l, \tau' |\ ,
\quad \sum_{k=N_-}^{ N_+ }\ \sum_{\sigma,\sigma'}\
\rho_{k \sigma\sigma', k\sigma\sigma'}=1\ ,
\label{4.4}
\end{equation}
where $N_-={\rm max}\{0,N-M+m\}$ and $N_+={\rm min}\{N,m\}$ are the minimum and maximum
number of fermions that the first partition can contain, due to the exclusion principle.
Using this decomposition, one can obtain a full characterization of the structure
of entangled fermion states (see \cite{Benatti6} for further details):

\begin{proposition}
A generic $(m, M-m)$-mode bipartite state (\ref{4.4}) is entangled
if and only if it can not be cast in the following block diagonal form
\begin{equation}
\rho_N=\sum_{k=N_-}^{ N_+ } p_k\ \rho_k\ ,\qquad 
\sum_{k=N_-}^{ N_+ } p_k=1\ ,\quad {\rm Tr}[\rho_k]=1\ ,
\label{28}
\end{equation}
with
\begin{equation}
\rho_k=\sum_{\sigma,\sigma',\tau,\tau'}\
\rho_{k \sigma\sigma', k\tau\tau'}\ | k, \sigma; N-k, \sigma'\rangle \langle k, \tau; N-k, \tau' |\ ,
\quad \sum_{\sigma,\sigma'}\rho_{k \sigma\sigma', k\sigma\sigma'}=1\ ,
\label{29}
\end{equation}
({\it i.e.} at least one of its non-diagonal coefficients $\rho_{k \sigma\sigma', l\tau\tau'}$, $k\neq l$,
is nonvanishing),
or, if it can, if and only if at least one of its diagonal blocks $\rho_k$ is non-separable.
\end{proposition}

The extension of this result to the case of Majorana fermions requires some care,
since, as we shall see in the coming sections, 
for real fermions the GNS-representations of the observable algebra $\cal A$
are in general reducible.

\section{Algebraic approach to Majorana fermions}

For a system made of real fermions, the structure of the
algebra of observables $\cal A$ turns out to be quite different from that of ${\cal A}_f$
characterizing complex fermions and
discussed in the previous Section. As in that case, we shall describe the 
Majorana observable algebra in terms of real mode operators $c_i$, $i=1,2,\ldots, N$,
with $c_i{}^\star = c_i$, satisfying the following anticommutation relations:%
\footnote{Although the number $N$ of Majorana modes can also be infinite,
for simplicity, hereafter we shall limit our considerations to the more physically relevant
case of finite $N$.}
\begin{equation}
\{c_i,\,c_j\}=2\,\delta_{ij}\ ,\qquad i,\, j=1,2,\ldots,N\ .
\label{5.1}
\end{equation}
The linear span of all products of these mode operators, together with the unit
element $c_0\equiv {\bf 1}$, form the (euclidean) complex Clifford algebra
${\cal C}_N(\mathbb{C})$  \cite{Gilbert}-\cite{Meinrenken}, 
which is the operator $C^\star$-algebra relevant
to describe Majorana fermion systems.
One can easily show that the monomials
\begin{equation}
(c_1)^{n_1}\, (c_2)^{n_2}\ldots (c_N)^{n_N}\ ,\qquad n_i=0,1\ ,
\label{5.2}
\end{equation}
with $(c_1)^0\, (c_2)^0\ldots (c_N)^0$ interpreted as the identity,
form a basis in ${\cal C}_N$, which therefore has dimension $2^N$.
Finite dimensional Clifford algebras are isomorphic to matrix algebras \cite{Gilbert};
however, one has to distinguish two cases, according to whether $N$ is even or odd.
When $N=2n$, the algebra ${\cal C}_{2n}(\mathbb{C})$ is isomorphic to the
$2^n\times 2^n$ matrix algebra ${\cal M}_{2^n}(\mathbb{C})$, while for
$N=2n+1$, the algebra ${\cal C}_{2n+1}(\mathbb{C})$ is isomorphic to the
direct sum ${\cal M}_{2^n}(\mathbb{C}) \oplus {\cal M}_{2^n}(\mathbb{C})$.
Explicitly, up to unitary equivalences, in the case ${\cal C}_{2n}$ the isomorphism is given by:
\begin{eqnarray}
\nonumber
&& c_{2k} \longleftrightarrow m_{2k}\equiv \underbrace{\sigma_0\otimes\ldots\otimes\sigma_0}_{n-k-1}
\otimes\,\sigma_1\otimes\underbrace{\sigma_3\otimes\ldots\otimes\sigma_3}_k \\
&& c_{2k+1} \longleftrightarrow m_{2k+1}\equiv \underbrace{\sigma_0\otimes\ldots\otimes\sigma_0}_{n-k-1}
\otimes\,\sigma_2\otimes\underbrace{\sigma_3\otimes\ldots\otimes\sigma_3}_k  \qquad k=1,2,\ldots,n
\label{5.3}
\end{eqnarray}
where $\sigma_i$, $i=1,2,3$ are the Pauli matrices, with $\sigma_0$ the $2\times2$ identity matrix,
while for ${\cal C}_{2n+1}$ one obtains:
\begin{eqnarray}
\nonumber
&& c_{k} \longleftrightarrow m_{2k}\oplus m_{2k}\ ,\qquad k=1,2,\ldots,2n \\
&& c_{2n+1} \longleftrightarrow \big(\underbrace{\sigma_3\otimes\ldots\otimes\sigma_3}_n\big)
\oplus\big(\underbrace{-\sigma_3\otimes\ldots\otimes-\sigma_3}_n\big)  \ .
\label{5.4}
\end{eqnarray}

Although out of $2n$ Clifford modes one can construct $n$ ordinary complex fermions modes
through the relations
\begin{eqnarray}
\nonumber
&& a_k=\frac{1}{2}\big(c_{2k-1} + i c_{2k}\big)\\
&& a_k^\star=\frac{1}{2}\big(c_{2k-1} - i c_{2k}\big)\ ,\qquad k=1,2,\ldots,n\ ,
\label{5.5}
\end{eqnarray}
the properties of the Clifford algebra ${\cal C}_N$ are quite different from those of the
fermion algebra ${\cal A}_f$. First of all, there is no exclusion principle
for real fermion modes, since $(c_i)^2=1$, and not zero, as for the complex fermion
operators $a_k$, $a_k^\star$ in (\ref{5.5}). In fact, one can not even speak
of occupancy of a Clifford mode, since there is no number operator in ${\cal C}_N$.%
\footnote{One can certainly form hermitian bilinears in the Clifford modes,
{\it e.g.} $i\, c_{2k-1}\, c_{2k}$, but these are related to the occupation number
operator of the corresponding complex fermion modes and not of the Clifford modes.}
In a sense, a Clifford mode is always filled and empty at the same time.
Recalling the discussion of the previous Section, this implies that the Clifford algebras
do not admit Fock representations; therefore, all the results regarding separability
and entanglement given before for standard fermions need to be reconsidered.

As in the case of ${\cal A}_f$, the automorphism $\vartheta$, defined by its action
on the mode operators as $\vartheta(c_i)=-c_i$, allows decomposing ${\cal C}_N$ in its
even ${\cal C}_N^e$ and odd ${\cal C}_N^o$ parts.
Then, following {\sl Definition 1},
a bipartition $({\cal A}_1,\, {\cal A}_2)$ of the Clifford algebra ${\cal C}_N$ 
is given by two Clifford subalgebras
${\cal A}_1,\ {\cal A}_2\subset{\cal C}_N$, having only the unit element in common, and such that
${\cal A}_1\cup{\cal A}_2={\cal C}_N$, 
together with $[{\cal A}_1^{e}\,,\,{\cal A}_2]=[{\cal A}_1\,,\,{\cal A}_2^{e}]=0$. 

In practice, the bipartition $({\cal A}_1,\, {\cal A}_2)$ 
is obtained by splitting the collection of modes $\{c_i\}_{i=1,2\ldots,N}$ into two disjoint sets
$\{c_i\, | i=1,2,\ldots,p\}$ and $\{c_j\, | j=p+1,p+2,\ldots,N\}$.
The linear span of the monomials $(c_1)^{n_1}\, (c_2)^{n_2}\ldots (c_p)^{n_p}$, with $n_i=0,1$,
gives the subalgebra ${\cal A}_1$, while that of
$(c_{p+1})^{n_{p+1}}\, (c_{p+2})^{n_{p+2}}\ldots (c_N)^{n_N}$ generate the subalgebra ${\cal A}_2$.
In practice, also in this case a bipartition of ${\cal C}_N$ is determined by the choice of the
integer $p$, with $0<p<N$.

As for any $C^\star$-algebra, a state on the Clifford algebra ${\cal C}_N$ is given by a positive linear map
from ${\cal C}_N$ to $\mathbb{C}$. The most simple state is given by the map $\Omega$ that sends
all elements of ${\cal C}_N$ to zero, except for the identity, which is mapped to one \cite{Price}-\cite{Alicki}; 
on the basis elements (\ref{5.2}), one then has:%
\footnote{For $N$ even, this state corresponds to a thermal state for the algebra ${\cal A}_f$
generated by the complex fermion modes in (\ref{5.5}), in the limit of infinite temperature.}
\begin{equation}
\Omega\big( (c_1)^{n_1}\, (c_2)^{n_2}\ldots (c_N)^{n_N} \big)=
\delta_{n_1,0}\,\delta_{n_2,0}\ldots \delta_{n_N,0}\ .
\label{5.6}
\end{equation}
Through the standard GNS-construction, from the couple $({\cal C}_N,\, \Omega)$ defining
Majorana fermion systems, one construct an Hilbert space ${\cal H}_\Omega$ and a representation
$\pi_\Omega$ on it; the space ${\cal H}_\Omega$ is generated by applying elements
of ${\cal C}_N$ to the cyclic vector $|\Omega\rangle$, so that for any element $\gamma\in {\cal C}_N$
one has:
\begin{equation}
\Omega(\gamma)=\langle\Omega|\pi_\Omega(\gamma) |\Omega\rangle\ .
\label{5.7}
\end{equation}
Such a state is separable with respect to the bipartition $({\cal A}_1,{\cal A}_2)$ defined above. 
Indeed, as already mentioned, the elements $\alpha^{(1)}\in{\cal A}_1$ and  
$\alpha^{(2)}\in{\cal A}_2$ have the generic form:
\begin{align}
\nonumber
&\alpha^{(1)}=\sum_{n_1,n_2,\ldots,n_p} \alpha^{(1)}_{n_1,n_2,\ldots,n_p}
\ (c_1)^{n_1}\, (c_2)^{n_2}\ldots (c_p)^{n_p}\ ,\\
&\alpha^{(2)}=\sum_{n_{p+1}, n_{p+2},\ldots n_N}\alpha^{(2)}_{n_{p+1},n_{p+2},\ldots,n_N}\
(c_{p+1})^{n_{p+1}}\, (c_{p+2})^{n_{p+2}}\ldots (c_N)^{n_N}\ ,
\nonumber
\end{align}
with $n_i=0,1$ and 
coefficients $\alpha^{(1)}_{n_1,n_2,\ldots,n_p},\, \alpha^{(2)}_{n_{p+1},n_{p+2},\ldots,n_N}\in {\mathbb{C}}$.
Then,
\begin{equation}
\Omega(\alpha^{(1)}\alpha^{(2)})=\alpha^{(1)}_{0,0,\ldots,0}\, \alpha^{(2)}_{0,0,\ldots,0}
=\Omega(\alpha^{(1)})\,\Omega(\alpha^{(2)})\ .
\label{5.7-1}
\end{equation}
\medskip
\noindent
{\bf Remark:} 
Due to the anticommutative character of the Clifford modes,
$c_i c_j=-c_j c_i$, for $i\neq j$, the restrictions on the form of the states
of a fermion algebra given in {\sl Lemma 2} hold also for the Clifford
algebra ${\cal C}_N$, as the entanglement criterion given in {\sl Corollary 3}.
\hfill$\Box$

\bigskip
A basis in the Hilbert space ${\cal H}_\Omega$ can be obtained by applying the
basis elements in (\ref{5.2}) to the cyclic vector:
\begin{equation}
|n_1, n_2,\ldots,n_M\rangle= 
\big[\pi_{\Omega}(c_1)\big]^{n_1}\, \big[\pi_{\Omega}(c_2)\big]^{n_2}\, \cdots\, 
\big[\pi_{\Omega}(c_N)\big]^{n_N}\,|\Omega\rangle\ ;
\label{5.8}
\end{equation}
these vectors are clearly orthogonal among themselves thanks to (\ref{5.6}).
This basis contains $2^N$ vectors, so that the GNS-representation
$\pi_\Omega$ turns out to be highly reducible.
For instance, when $N=2n$, any element $\gamma\in {\cal C}_N$ will be represented
in $\pi_\Omega$ by a $2^{2n}\times 2^{2n}$ matrix, {\it i.e.} by elements
of ${\cal M}_{4^{n}}$, while, as explicitly shown by (\ref{5.3}), ${\cal C}_N$
is isomorphic to ${\cal M}_{2^n}$.

In order to study separability and entanglement in the case of reducible
GNS-represen-tations, one needs to generalize the treatment presented in Section 3,
which is appropriate only for irreducible GNS-representations, as the Fock
representation used to discuss standard fermions.
We shall see that reducibility allows for richer structures in the classification
scheme of entangled Majorana states.

\section{Reducible GNS-representations}

In order to properly treat reducible GNS-representations, one needs to generalize the
algebraic approach to quantum systems presented in Section 2.
As we have seen, for any quantum system defined by the operator algebra $\cal A$ and
a state~$\Omega$, the {\sl GNS-construction} allows building a triple
$\big({\cal H}_\Omega,\ \pi_\Omega,\ |\Omega\rangle\big)$, 
so that the system can be described in terms of bounded operators 
$\pi_\Omega({\cal A})\subset {\cal B}({\cal H}_\Omega)$ acting
on the Hilbert space ${\cal H}_\Omega$ spanned by the (completion of the)
set of vectors $\big\{ \pi_\Omega({\cal A}) |\Omega\rangle \big\}$.

When the representation $\pi_\Omega$ is not irreducible, as for any algebra
representation, it can be always decomposed 
into irreducible representations $\pi_\Omega^{(\mu,r)}$,
\begin{equation}
\pi_\Omega=\oplus_{\mu,r}\ \pi_\Omega^{(\mu,r)}\ .
\label{6.1}
\end{equation}
Two indices $\mu$ and $r$ will be used to label such representations:
the greek index $\mu$ distinguishes among different irreducible representations, while
the latin index $r$ counts the multiplicity of a given irreducible representation.
In other terms, $\pi_\Omega^{(\mu,r)}$ and $\pi_\Omega^{(\nu,s)}$, with $\mu\neq\nu$,
are different irreducible representations, while $\pi_\Omega^{(\mu,r)}$ and $\pi_\Omega^{(\mu,s)}$,
with $r\neq s$, are two copies of the same irreducible representation $\pi_\Omega^{(\mu)}$.
We shall call $d_\mu$ and $m_\mu$ the dimension and the multiplicity of $\pi_\Omega^{(\mu)}$.

\medskip
\noindent
{\bf Remark:} Notice that, contrary to the usual convention, in the decomposition (\ref{6.1}) 
unitarily equivalent representations are treated as distinct. This is necessary in discussing
quantum separability issues, since unitary transformations might map a given bipartition
$({\cal A}_1,\, {\cal A}_2)$ of $\cal A$ into a different one \cite{Benatti1,Benatti6}. \hfill$\Box$
\medskip

\noindent
To the decomposition of representations as in (\ref{6.1}) there corresponds a similar decomposition
of the Hilbert space ${\cal H}_\Omega$:
\begin{equation}
{\cal H}_\Omega=\oplus_{\mu,r}\ {\cal H}_\Omega^{(\mu,r)}\ ,
\label{6.2}
\end{equation}
so that for any element $\alpha\in{\cal A}$, the operator $\pi_\Omega^{(\mu,r)}(\alpha)$ acts
nontrivially only on the irreducible subspace ${\cal H}_\Omega^{(\mu,r)}$.
Let $\big\{ |e_i^{(\mu,r)}\rangle\, |\ i=1,2,\ldots, 
d_\mu\big\}$
be a set of elements of ${\cal H}_\Omega$ forming an orthonormal basis for the subspace
${\cal H}_\Omega^{(\mu,r)}$. Since the GNS-representation $\pi_\Omega^{(\mu,r)}$ is irreducible,
these states are pure, as any element of this subspace,
and the whole ${\cal H}_\Omega^{(\mu,r)}$ can be obtained 
by applying the operators $\pi_\Omega({\cal A})$ to the
normalized cyclic vector:
\begin{equation}
|\Omega^{(\mu,r)}\rangle={1\over \sqrt{N^{(\mu,r)}}}\sum_i\, \langle e_i^{(\mu,r)}|\Omega\rangle\ 
|e_i^{(\mu,r)}\rangle\ ,\qquad N^{(\mu,r)}=\sum_i\big| \langle e_i^{(\mu,r)}|\Omega\rangle \big|^2\ .
\label{6.3}
\end{equation}

On the other hand, a generic element in the full Hilbert space ${\cal H}_\Omega$
turns out to be in general a mixed state when restricted to the 
operator algebra $\pi_\Omega({\cal A})$.%
\footnote{Although it is surely a pure state for the full algebra ${\cal B}({\cal H}_\Omega)$ of
bounded operators on ${\cal H}_\Omega$.}
Indeed, any normalized state $|\psi\rangle\in{\cal H}_\Omega$ can be expanded using the collection
of basis elements $\big\{ |e_i^{(\mu,r)}\rangle \big\}$ introduced above as:

\begin{equation}
|\psi\rangle=\sum_{\mu,r,i}\, \langle e_i^{(\mu,r)}|\psi\rangle\,
|e_i^{(\mu,r)}\rangle\ , \qquad 
\sum_{\mu,r,i}\, \big|\langle e_i^{(\mu,r)}|\psi\rangle\big|^2=1\ .
\label{6.4}
\end{equation}
Further, thanks to the irreducibility of the representations $\pi_\Omega^{(\mu,r)}$, one has:
\begin{equation}
\langle e_i^{(\mu,r)}| \pi_\Omega(\alpha) | e_j^{(\nu,s)}\rangle=
\delta_{\mu,\nu}\,\delta_{r,s}\, \big[\pi_\Omega^{(\mu)}(\alpha)\big]_{ij}\ ,
\label{6.5}
\end{equation}
where $\big[\pi_\Omega^{(\mu)}(\alpha)\big]_{ij}$ is the matrix representation of
the element $\alpha\in{\cal A}$ in the irreducible representation $\pi_\Omega^{(\mu)}$.
In fact, recall that $\pi_\Omega^{(\mu,r)}$, with $r=1,2,\ldots, m_\mu$, 
are all copies of the same representation
$\pi_\Omega^{(\mu)}$; thus, the matrix elements 
$$
\big[\pi_\Omega^{(\mu)}(\alpha)\big]_{ij}\equiv
\langle e_i^{(\mu,r)}| \pi_\Omega(\alpha) | e_i^{(\mu,r)}\rangle\ ,
$$
are actually independent
from the multiplicity index $r$, or equivalently, the representation matrix $\big[\pi_\Omega^{(\mu)}(\alpha)\big]$
is the same for all the $m_\mu$ copies $\pi_\Omega^{(\mu,r)}$, $r=1,2,\dots,m_\mu$.
As a consequence, the mean value of any element $\alpha\in{\cal A}$ with respect to the state $|\psi\rangle$
can be represented by means of a density matrix $\rho_\psi$, using the trace operation:
\begin{equation}
\langle\psi | \pi_\Omega(\alpha) |\psi\rangle={\rm Tr}\big[\rho_\psi\, \pi_\Omega(\alpha)\big]\ ,
\label{6.6}
\end{equation}
where
\begin{equation}
\rho_\psi=\sum_\mu\sum_{ij}\, \lambda_{ij}^{(\mu)}\, | e_i^{(\mu)}\rangle\langle e_j^{(\mu)}|\ ,
\label{6.7}
\end{equation}
with $\lambda_{ij}^{(\mu)}=\sum_r \langle\psi| e_i^{(\mu,r)}\rangle\langle e_j^{(\mu,r)} |\psi\rangle$,
and $\{ | e_i^{(\mu)}\rangle \}$ any basis in ${\cal H}_\Omega$ carrying the 
irreducible representation $\pi_\Omega^{(\mu)}$; in practice, a convenient choice
for $\{ | e_i^{(\mu)}\rangle \}$
is the basis $\big\{ |e_i^{(\mu,r)}\rangle \big\}$ in ${\cal H}_\Omega^{(\mu,r)}$, with any fixed index $r$,
since, as remarked above, each one of these spaces carries the same irreducible representation
$\pi_\Omega^{(\mu)}$ of $\cal A$.

In particular, the cyclic GNS-vector $|\Omega\rangle$ turns out to be represented by
the density matrix $\rho_\Omega$ of the general form (\ref{6.7}), with
$\lambda_{ij}^{(\mu)}=\sum_r \langle\Omega| e_i^{(\mu,r)}\rangle\langle e_j^{(\mu,r)} |\Omega\rangle$,
so that, for any $\alpha\in{\cal A}$:
\begin{equation}
\Omega(\alpha)={\rm Tr}\big[\rho_\Omega\, \pi_\Omega(\alpha)\big]\ .
\label{6.8}
\end{equation}
Similarly, any mixed state $\rho$ on ${\cal H}_\Omega$, which in general can be decomposed as
\begin{equation}
\rho=\sum_{\mu,\nu,r,s,i,j} \lambda_{ij}^{(\mu,r;\nu,s)}\ 
| e_i^{(\mu,r)}\rangle\langle e_j^{(\nu,s)}|\ ,\qquad \sum_{\mu,r,i} \lambda_{ii}^{(\mu,r;\mu,r)}=1\ ,
\label{6.9}
\end{equation}
when restricted to the algebra $\pi_\Omega(\cal A)$, also reduces to the generic form (\ref{6.7}).

There are however notable exceptions to this general rule. Let us fix the irreducible representation
$\pi_\Omega^{(\mu)}$ and consider the following linear combination in ${\cal H}_\Omega$
\begin{equation}
| f_i^{(\mu,r)}\rangle=\sum_{s,j} U_{rs}\, V_{ij}\, | e_j^{(\mu,s)}\rangle\ ,
\label{6.10}
\end{equation}
with $U$ and $V$ unitary matrices and 
$\big\{ | e_i^{(\mu,r)}\rangle\, |\, i=1,2,\ldots,d_\mu \big\}$ orthonormal basis in 
${\cal H}_\Omega^{(\mu,r)}$, $r=1,2,\dots,m_\mu$, the $m_\mu$ Hilbert subspaces
carrying the representation $\pi_\Omega^{(\mu)}$.
When restricted to $\pi_\Omega({\cal A})$, the vectors $| f_i^{(\mu,r)}\rangle$ behave as the linear
combinations $|\tilde e_i^{(\mu,r)}\rangle=\sum_j V_{ij}\, | e_j^{(\mu,r)}\rangle$, since, due to the identity
(\ref{6.5}) above, the following matrix elements
\begin{eqnarray}
\nonumber
\langle f_i^{(\mu,r)}| \pi_\Omega(\alpha) | f_j^{(\mu,r)}\rangle &=&
\sum_{k,\ell,p,q} U_{rq}\, U^\dagger_{pr}\,  V_{ki}^\dagger\, V_{j\ell}\
\langle e_k^{(\mu,p)}| \pi_\Omega(\alpha) |  e_\ell^{(\mu,q)}\rangle\\
&=&\sum_{k,\ell} V_{ki}^\dagger\, V_{j\ell}\, \big[\pi_\Omega^{(\mu)}(\alpha)\big]_{k\ell}
\equiv\langle \tilde e_i^{(\mu,r)}| \pi_\Omega(\alpha) | \tilde e_j^{(\mu,r)}\rangle\ ,
\label{6.11}
\end{eqnarray}
are actually independent from the index $r$.
Being combinations of basis states, 
the vectors in $\big\{ | \tilde e_i^{(\mu,r)}\rangle\, |\, i=1,2,\ldots,d_\mu \big\}$ 
are pure, forming another orthonormal basis 
in ${\cal H}_\Omega^{(\mu,r)}$; as a consequence, also the more general combinations $| f_i^{(\mu,r)}\rangle$ 
in (\ref{6.10}) represent pure states on the subalgebra $\pi^{\mu,r}_\Omega({\cal A})$. This result will be important for
the discussions that will follow.

\medskip
\noindent
{\bf Remark:} The vectors in (\ref{6.10}) forming the set 
$\big\{ | f_i^{(\mu,r)}\rangle\, |\, i=1,2,\ldots,d_\mu \big\}$
are clearly orthonormal and, as shown by (\ref{6.11}), span an invariant subspace of ${\cal H}_\Omega$, 
which is however different from ${\cal H}_\Omega^{(\mu,r)}$; it carries
a representation of $\cal A$ unitarily equivalent to $\pi_\Omega^{(\mu)}$,
coinciding with it only when $V={\bf 1}$.
This means that the partial decomposition 
$\oplus_r {\cal H}_\Omega^{(\mu,r)}$ into subspaces carrying the representation 
$\pi_\Omega^{(\mu)}$ is in general not unique.%
\footnote{This might have some consequences when evaluating
the von Neumann entropy of the state $\Omega$ \cite{Balachandran2,Balachandran3}.}
\hfill$\Box$

\section{Reducibility and entanglement}

When a state $\Omega$ for the operator algebra $\cal A$ gives rise to a reducible
GNS-representation $\pi_\Omega$, the analysis of the notions of separability and entanglement
according to {\sl Definition 3} in Section 3 becomes more involved. Following the previous
discussion, one can decompose $\pi_\Omega$ into its irreducible components
\begin{equation}
\pi_\Omega=\oplus_{\mu,r}\ \pi_\Omega^{(\mu,r)}\ ,
\label{7.1}
\end{equation}
where $\pi_\Omega^{(\mu,r)}$, for $\mu$ fixed and $r=1,2,\ldots,m_\mu$, are $m_\mu$ copies
of the same irreducible representation $\pi_\Omega^{(\mu)}$. Correspondingly, one has a similar
decomposition for the GNS-Hilbert space, ${\cal H}_\Omega=\oplus_{\mu,r}\ {\cal H}_\Omega^{(\mu,r)}$,
where, for $\mu$ fixed, the $m_\mu$ subspaces ${\cal H}_\Omega^{(\mu,r)}$ are all isomorphic,
and, without loss of generality, they can be identified.

Let us now fix a bipartition $({\cal A}_1, {\cal A}_2)$ of $\cal A$ and consider an orthonormal basis
$\big\{ |e_i^{(\mu,r)}\rangle \big\}$ in each Hilbert space
${\cal H}_\Omega^{(\mu,r)}$; these states are pure since they carry the irreducible
representation $\pi_\Omega^{(\mu)}$. In addition, they can be chosen to be separable:

\begin{lemma}
Given any bipartition $({\cal A}_1, {\cal A}_2)$ of the algebra $\cal A$, and a separable state $\Omega$ 
leading to the reducible representation $\pi_\Omega$ with decomposition
as in (\ref{6.1}) and (\ref{6.2}), it is always possible
to select in ${\cal H}_\Omega^{(\mu,r)}$ an orthonormal basis $\big\{ |e_i^{(\mu,r)}\rangle \big\}$
of separable pure states.
\end{lemma}

\begin{proof}
The statement can be proven by explicitly constructing the basis.
For simplicity, we shall consider $\cal A$ to be a boson algebra; however, using the techniques
presented in Section~4, the proof can be
easily extended to the fermion case.
The building procedure involves selecting two selfadjoint elements $\alpha_1\in {\cal A}_1$
and $\alpha_2\in {\cal A}_2$, $\alpha_i^\star=\alpha_i$, in the two partitions.
On the space ${\cal H}_\Omega^{(\mu,r)}$, these elements are represented by the hermitian
operators $\pi_\Omega^{(\mu,r)}(\alpha_i)$, $i=1,2$, with spectral decomposition:
\begin{equation}
\pi_\Omega^{(\mu,r)}(\alpha_1)=\sum_k \alpha_k^{(1)}\, P_k^{(\mu,r)}\ ,\qquad
\pi_\Omega^{(\mu,r)}(\alpha_2)=\sum_\ell \alpha_\ell^{(2)}\, Q_\ell^{(\mu,r)}\ ,\qquad 
\alpha_k^{(1)},\ \alpha_\ell^{(2)}\in\mathbb{R}\ .
\label{7.2}
\end{equation}
Since the GNS-vector $|\Omega\rangle\in{\cal H}_\Omega$ is assumed to be separable,
by acting on it with the projectors $P_k^{(\mu,r)}\in\pi_\Omega^{(\mu,r)}\big({\cal A}_1\big) $ and 
$Q_\ell^{(\mu,r)}\in\pi_\Omega^{(\mu,r)}\big({\cal A}_2\big)$
one builds a basis of manifestly separable pure states of the form:%
\footnote{We are here assuming that the projectors $P_k^{(\mu,r)}$ and $Q_\ell^{(\mu,r)}$ 
correspond to elements 
belonging to the algebra $\cal A$ and, as we shall see, this is indeed the case when $\cal A$ is
a Clifford algebra. However, in more general cases, this condition might not hold;
in such instances, one simply applies all considerations to the von Neumann algebra extension
of $\cal A$ \cite{Bratteli}.}
\begin{equation}
|e_i^{(\mu,r)}\rangle=\frac{1}{\langle \Omega| P_k^{(\mu,r)} |\Omega\rangle\,
\langle \Omega| Q_\ell^{(\mu,r)} |\Omega\rangle}\
P_k^{(\mu,r)}\, Q_\ell^{(\mu,r)}\ |\Omega\rangle\ ,\qquad i\equiv(k,\ell)\ .
\label{7.3}
\end{equation}
These states satisfy the separability condition (\ref{3.4}) and thus,
by taking in it $\alpha_1=\alpha_2={\bf 1}$, they are orthonormal:
\begin{equation}
\langle e_i^{(\mu,r)}|e_{i'}^{(\mu,r)}\rangle=
\frac{\langle \Omega|P_k^{(\mu,r)}\,P_{k'}^{(\mu,r)}\Omega\rangle\
\langle\Omega| Q_\ell^{(\mu,r)}\, Q_{\ell'}^{(\mu,r)} |\Omega\rangle}
{\langle \Omega| P_k^{(\mu,r)} |\Omega\rangle\,
\langle \Omega| Q_\ell^{(\mu,r)} |\Omega\rangle\,
\langle \Omega| P_{k'}^{(\mu,r)} |\Omega\rangle\,
\langle \Omega| Q_{\ell'}^{(\mu,r)} |\Omega\rangle}
=\delta_{k k'}\, \delta_{\ell \ell'}
\equiv\delta_{i i'}\ .
\label{7.4}
\end{equation}
Furthermore, the set of vectors (\ref{7.3}) form a basis for the space ${\cal H}_\Omega^{(\mu,r)}$.
Indeed, the existence of an element $|\phi\rangle\in {\cal H}_\Omega^{(\mu,r)}$
not belonging to the span of the set $\big\{ |e_i^{(\mu,r)}\rangle \big\}$ would be in
contradiction with the assumption that ${\cal A}_1 \cup {\cal A}_2={\cal A}$; in fact,
by construction,
the set $\big\{ P_k^{(\mu,r)} \big\}$ generates $\pi_\Omega^{(\mu,r)}\big({\cal A}_1\big)$,
while $\big\{ Q_\ell^{(\mu,r)} \big\}$ generates $\pi_\Omega^{(\mu,r)}\big({\cal A}_2\big)$.
\end{proof}

\medskip
\noindent
{\bf Remark:}
If some of the above projectors turn out to annihilate $|\Omega\rangle$, {\it e.g.}
\hbox{$P_k^{(\mu,r)}|\Omega\rangle=0$}, for some $k$, one considers the subalgebra generated by them
and repeats the previous construction by choosing a suitable selfadjoint element in it. 
For a finitely generated algebra $\cal A$,
the successive application of this procedure will surely come to an end and provide the
wanted separable basis $\big\{ |e_i^{(\mu,r)}\rangle \big\}$. 
\hfill$\Box$

\medskip

Having constructed in each space ${\cal H}_\Omega^{(\mu,r)}$ a basis 
$\big\{ |e_i^{(\mu,r)}\rangle \big\}$ of vectors that result separable with respect to the chosen bipartition, one can now consider arbitrary linear combinations of
these vectors. In general, such combinations will no longer be separable.
For instance, even limiting the attention to a single space ${\cal H}_\Omega^{(\mu,r)}$,
with fixed indices $\mu$ and $r$,
the following combinations of vectors 
$|\tilde e_i^{(\mu,r)}\rangle=\sum_j V_{ij}\, | e_j^{(\mu,r)}\rangle$,
with $V_{ij}$ arbitrary complex coefficients, are still pure states in ${\cal H}_\Omega^{(\mu,r)}$,
as discussed in the previous Section; however,
they are no longer separable, since in general
the expectation $\langle e_i^{(\mu,r)} | \pi_\Omega^{(\mu,r)}(\alpha_1)\,
\pi_\Omega^{(\mu,r)}(\alpha_2) |e_i^{(\mu,r)}\rangle$ of any local operator $\alpha_1\alpha_2$,
$\alpha_i\in{\cal A}_i$, $i=1,2$, can not be written in product form as in (\ref{3.4}).

Nevertheless, there are linear combinations involving basis vectors in spaces 
${\cal H}_\Omega^{(\mu,r)}$ with different index $r$ that remain separable.

\begin{lemma}
Within the hypothesis of the previous Lemma, let us consider the following
linear combinations of basis states:
\begin{equation}
| g_i^{(\mu,r)}\rangle=\sum_s U_{rs}\, | e_i^{(\mu,s)}\rangle\ ,
\label{7.5}
\end{equation}
with $U$ a unitary matrix. These states are pure and separable.
\end{lemma}

\begin{proof}
As already shown, 
the matrix elements of any operator $\pi_\Omega(\alpha)$, $\alpha\in{\cal A}$,
with respect to the vectors of the set $\big\{| g_i^{(\mu,r)}\rangle\big\}$ coincide with those of the corresponding vectors 
in the set $\big\{| e_i^{(\mu,r)}\rangle\big\}$, since both set of vectors carry the same irreducible representation
$\pi_\Omega^\mu$ of $\cal A$ (see (\ref{6.11}) with $V={\bf 1}$). 
Then, since the separability condition (\ref{3.4}) holds by construction for the
elements of the basis $\big\{| e_i^{(\mu,r)}\rangle\big\}$, it is automatically true also for
the vectors in $\big\{| g_i^{(\mu,r)}\rangle\big\}$.
\end{proof}

More in general, a state on the algebra $\cal A$ that belongs to the {\sl folium} of $\Omega$ is mixed, 
and thus represented by a density
matrix $\rho$ on ${\cal H}_\Omega$. It can be decomposed as in (\ref{6.7}),
\begin{equation}
\rho=\sum_\mu\sum_{i,j}\, \lambda_{ij}^{(\mu)}\, | e_i^{(\mu)}\rangle\langle e_j^{(\mu)}|\ ,\qquad
\sum_{\mu,i} \lambda_{ii}^{(\mu)}=1\ ,
\label{7.6}
\end{equation}
where the set $\big\{| e_i^{(\mu)}\rangle\big\}$ is a separable basis in ${\cal H}_\Omega$ 
carrying the irreducible representation
$\pi_\Omega^{(\mu)}$; in practice, as mentioned earlier, it can be taken to coincide with any separable basis
$\big\{| e_i^{(\mu,r)}\rangle\big\}$ in ${\cal H}_\Omega^{(\mu,r)}$ introduced above, with arbitrary,
but fixed $r$.

It follows that a state in diagonal form,
\begin{equation}
\rho_D=\sum_\mu\sum_{i}\, \lambda_{ii}^{(\mu)}\, | e_i^{(\mu)}\rangle\langle e_i^{(\mu)}|\ ,
\label{7.7}
\end{equation}
is surely separable, being the convex combination of separable, rank-1 projectors.
One can then conclude that:

\begin{proposition}
A generic mixed state $\rho=\sum_\mu \rho^{(\mu)}$ as in (\ref{7.6}) is entangled with respect to the
given bipartition $({\cal A}_1, {\cal A}_2)$ if and only if at least
one of its irreducible components $\rho^{(\mu)}$,
\begin{equation}
\rho^{(\mu)}=\frac{1}{\lambda^{(\mu)}}\sum_{i,j}\, \lambda_{ij}^{(\mu)}\, | e_i^{(\mu)}\rangle\langle e_j^{(\mu)}|\ ,
\qquad \lambda^{(\mu)}\equiv\sum_i \lambda_{ii}^{(\mu)}\ ,
\label{7.8}
\end{equation}
results non separable.
\end{proposition}
As a consequence, the study of quantum correlations in the reducible representation $\pi_\Omega$
of the algebra $\cal A$, as given by the state
$\Omega$ through the {\sl GNS-construction}, reduces to the analysis
of entanglement in each of its irreducible components $\pi_\Omega^{(\mu)}$,
for which the results given in Section 3 apply.

All the above discussion can be made very explicit in the case of Majorana fermion systems,
{\it i.e.} when the operator algebra $\cal A$ coincides with the Clifford algebra ${\cal C}_N$
and for $\Omega$ the state introduced in (\ref{5.6}) is chosen.

\section{Structure of entangled Majorana states: ${\cal C}_2$}

We shall start discussing the simplest Majorana system,
the one defined by the operator algebra ${\cal C}_2$, generated by
the two mode elements $c_1$ and $c_2$. As discussed
in Section 5, the entire Clifford algebra ${\cal C}_2$ is then obtained as the linear span of the
following four basis elements: $\{ {\bf 1},\, c_1,\, c_2,\, c_1 c_2 \}$. As state 
$\Omega$ on this algebra,
we shall choose the one given in (\ref{5.6}), so that:
$\Omega(c_1)=\Omega(c_2)=\Omega(c_1 c_2)=\,0$, while $\Omega({\bf 1})=1$.

Given the state $\Omega$, the {\sl GNS-construction} provides a representation $\pi_\Omega$ of ${\cal C}_2$ on
a four-dimensional Hilbert space ${\cal H}_\Omega$, which is given by the linear span
of the four vectors obtained by applying the basis elements ${\bf 1}$, $c_1$, $c_2$ and $c_1 c_2$
to the cyclic vector $|\Omega\rangle$.%
\footnote{For sake of simplicity, here and in the following we shall use the
same symbol to indicate the elements $c_i$ of the abstract Clifford algebra
and its corresponding GNS-representation $\pi_\Omega(c_i)$ as operators
acting on ${\cal H}_\Omega$.}
Since, as discussed in Section~5, ${\cal C}_2$ is isomorphic to ${\cal M}_2(\mathbb{C})$,
the algebra of $2\times 2$ complex matrices, the four-dimensional GNS-representation
$\pi_\Omega$ is reducible: it can be decomposed as
\begin{equation}
\pi_\Omega=\pi_\Omega^{(1)}\oplus\pi_\Omega^{(2)}\ ,
\label{8.1}
\end{equation}
in terms of two, equal, two-dimensional representations $\pi_\Omega^{(r)}$, acting on
two-dimensional Hilbert subspaces ${\cal H}_\Omega^{(r)}$, $r=1,2$ such that
${\cal H}_\Omega={\cal H}_\Omega^{(1)} \oplus {\cal H}_\Omega^{(2)}$.%
\footnote{In this situation, all irreducible representations are equal,
so that the index $\mu$ takes only one value and can be suppressed;
thus, in the decomposition (\ref{8.1}) only the multiplicity index $r=1,2$ appears.}
As a consequence, the chosen state $\Omega$ is not a pure state for ${\cal C}_2$.

The only non trivial bipartition $({\cal A}_1,{\cal A}_2)$ of the algebra ${\cal C}_2$
is the one in which the subalgebra ${\cal A}_1$ is the linear span of $\{ {\bf 1},\, c_1\}$, while
the subalgebra of ${\cal A}_2$ that of $\{ {\bf 1},\, c_2\}$. For the construction
of two orthonormal basis in ${\cal H}_\Omega^{(r)}$ formed by separable, 
pure states we follow the general
scheme outlined in the proof of {\sl Lemma 3} of the previous Section. 
The procedure involves choosing two generic
selfadjoint elements, $\alpha=a_0+a_1\, c_1$ in ${\cal A}_1$ and
$\beta=b_0+b_1\, c_2$ in ${\cal A}_2$, with $a_i,b_i\in \mathbb{R}$. 
Their spectral decomposition,
\begin{equation}
\begin{array}{ll}
\alpha=\big(a_0 + a_1\big)\, P_+ + \big(a_0 - a_1\big)\, P_-   & \hskip 1cm P_\pm= (1\pm c_1)/2\ , \\
& \\
\beta=\big(b_0 + b_1\big)\, Q_+ + \big(b_0 - b_1\big)\, Q_-   & \hskip 1cm Q_\pm= (1\pm c_2)/2\ ,
\end{array}
\label{8.2}
\end{equation}
allows constructing the following four orthonormal vectors
\begin{equation}
|e_i^{(r)}\rangle= \frac{1}{2}\, \hat e_i^{(r)}(c_1,c_2)\, |\Omega\rangle\ ,\qquad r=1,2\ ,\quad i=1,2\ ,
\label{8.3}
\end{equation}
where $\hat e_i^{(r)}(c_1,c_2)$ are (suitably normalized) products of the projectors $P_\pm$ and $Q_\pm$:
\begin{equation}
\left\{\begin{aligned}
&\hat e_1^{(1)}(c_1,c_2)=(1+c_1)(1+c_2) \\
&\hat e_2^{(1)}(c_1,c_2)=(1-c_1)(1+c_2)
\end{aligned}\right.
\hskip 2cm
\left\{\begin{aligned}
&\hat e_1^{(2)}(c_1,c_2)=(1+c_1)(1-c_2) \\
&\hat e_2^{(2)}(c_1,c_2)=(1-c_1)(1-c_2)\ .
\end{aligned}\right.
\label{8.4}
\end{equation}
One can easily check that the set $\big\{ |e_i^{(1)}\rangle \, | \, i=1,2\big\}$ is a basis for
the subspace ${\cal H}_\Omega^{(1)}\subset{\cal H}_\Omega$ carrying the irreducible representation
$\pi_\Omega^{(1)}$ for which
\begin{eqnarray}
\nonumber
&& {\bf 1} \longrightarrow \sigma_0\\
\label{8.5}
&& c_{1} \longrightarrow \sigma_3 \\
\nonumber
&& c_{2} \longrightarrow \sigma_1\ ,
\end{eqnarray}
and consequently $c_1 c_2 \longrightarrow i\sigma_2$.
Similarly, the set \hbox{$\big\{ |e_i^{(2)}\rangle \, | \, i=1,2\big\}$} is a basis in 
${\cal H}_\Omega^{(2)}\subset{\cal H}_\Omega$,
carrying another copy of the same irreducible representation. In view of this, as discussed before, 
the four states $|e_i^{(r)}\rangle$ are pure. Furthermore, they are manifestly separable with respect to the given bipartition;
indeed, one can explicitly check that they satisfy the condition (\ref{3.4})
for any local operator in ${\cal C}_2$.

The decomposition ${\cal H}_\Omega={\cal H}_\Omega^{(1)} \oplus {\cal H}_\Omega^{(2)}$ is however
not unique, due to the fact that the representation (\ref{8.5}) has multiplicty 2.
In fact, as discussed in the final {\sl Remark} of Section 6, the linear combinations
\begin{equation}
| g_i^{(r)}\rangle=\sum_{s=1}^2 U_{rs}\, | e_i^{(s)}\rangle\ ,
\label{8.6}
\end{equation}
with $U$ unitary, define two orthonormal basis $\big\{ |g_i^{(1)}\rangle \, | \, i=1,2\big\}$
and $\big\{ |g_i^{(2)}\rangle \, | \, i=1,2\big\}$, spanning two subspaces 
$\widetilde{\cal H}_\Omega^{(r)}\subset{\cal H}_\Omega$
giving a new decomposition 
${\cal H}_\Omega=\widetilde{\cal H}_\Omega^{(1)} \oplus \widetilde{\cal H}_\Omega^{(2)}$ of the GNS-Hilbert space.
However, as already noticed in the general case, the irreducible representations $\tilde\pi_\Omega^{(r)}$
of ${\cal C}_2$ corresponding to this new decomposition of ${\cal H}_\Omega$ coincide with old ones,
{\it i.e.} one has: $\tilde\pi_\Omega^{(1)}=\tilde\pi_\Omega^{(2)}\equiv \pi_\Omega^{(1)}=\pi_\Omega^{(2)}$.
In addition, by {\sl Lemma 4}, the new pure basis vectors $|g_i^{(r)}\rangle$ still represent separable states.

These considerations became very explicit by taking for instance $U=(\sigma_1+\sigma_3)/\sqrt2$, so that the
two basis states result:
\begin{equation}
\left\{\begin{aligned}
&|g_1^{(1)}\rangle=\displaystyle{\frac{1}{\sqrt2}}(1+c_1)\, |\Omega\rangle \\
&|g_2^{(1)}\rangle=\displaystyle{\frac{1}{\sqrt2}}(1-c_1)c_2\, |\Omega\rangle
\end{aligned}\right.
\hskip 2cm
\left\{\begin{aligned}
&|g_1^{(2)}\rangle=\displaystyle{\frac{1}{\sqrt2}}(1+c_1)c_2\, |\Omega\rangle \\
&|g_2^{(2)}\rangle=\displaystyle{\frac{1}{\sqrt2}}(1-c_1)\, |\Omega\rangle\ .
\end{aligned}\right.
\label{8.7}
\end{equation}
These states are manifestly separable and give rise to two copies of the same matrix representation
of ${\cal C}_2$ given in (\ref{8.5}).

On the other hand, if one instead considers as in (\ref{6.10}) linear unitary combinations
of the vectors $|e_i^{(r)}\rangle$ involving also the lower index,
\begin{equation}
| f_i^{(r)}\rangle=\sum_{s,j} U_{rs}\, V_{ij}\, | e_j^{(s)}\rangle\ , \qquad U U^\dagger={\bf 1}=V V^\dagger\ ,
\label{8.8}
\end{equation}
the resulting sets of vectors $\big\{ |f_i^{(1)}\rangle \, | \, i=1,2\big\}$ and
$\big\{ |f_i^{(2)}\rangle \, | \, i=1,2\big\}$ are still basis carrying the ${\cal C}_2$
irreducible representations $V\pi_\Omega^{(r)}V^\dagger$, $r=1,2$, unitarily equivalent to the
original ones, but the pure states $|f_i^{(r)}\rangle$ are in general no longer separable.
An interesting example, that will turn useful in the following, is given by the choice:
\begin{equation}
U=\frac{1}{2}\begin{pmatrix} 1-i & 1+i\\ 1+i & 1-i\end{pmatrix}\ ,\qquad
V=\frac{1}{\sqrt2}\begin{pmatrix} 1 & i\\ 1 & -i\end{pmatrix}\ ,
\label{8.9}
\end{equation}
giving rise to the basis vectors
\begin{equation}
|f_i^{(r)}\rangle= \frac {1}{\sqrt2}\, \hat f_i^{(r)}(c_1,c_2)\, |\Omega\rangle\ ,
\qquad r=1,2\ ,\quad i=1,2\ ,
\label{8.10}
\end{equation}
with
\begin{equation}
\left\{\begin{aligned}
&\hat f_1^{(1)}(c_1,c_2)=(c_1+c_2) \\
&\hat f_2^{(1)}(c_1,c_2)=(1+ic_1c_2)
\end{aligned}\right.
\hskip 1.5cm
\left\{\begin{aligned}
&\hat f_1^{(2)}(c_1,c_2)=(1-ic_1c_2) \\
&\hat f_2^{(2)}(c_1,c_2)=(c_1-c_2)
\end{aligned}\right.
\label{8.11}
\end{equation}
Using the entanglement criterion given in {\sl Corollary 3}, one immediately sees that
the vectors (\ref{8.10}) are non separable, since
$\langle f_i^{(r)} |c_1c_2 |f_i^{(r)}\rangle\neq\, 0$.

Coming now to the GNS state $|\Omega\rangle$, one can easily sees that, 
although generating the whole Hilbert space ${\cal H}_\Omega$,
it is not a pure state on the Clifford algebra ${\cal C}_2$. 
In fact, recalling (\ref{8.7}), one can write:
\begin{equation}
|\Omega\rangle=\frac{1}{\sqrt2} \big( |g_1^{(1)}\rangle + |g_2^{(2)}\rangle \big)\ .
\label{8.12}
\end{equation}
Since $\langle g_i^{(1)} |\alpha|g_j^{(2)}\rangle=0$, $i,j=1,2$, for any element $\alpha\in{\cal C}_2$,
due to the irreducibility of the representations carried by $\big\{ |g_i^{(1)}\rangle\big\}$ and
$\big\{ |g_i^{(2)}\rangle\big\}$, the mean value $\Omega(\alpha)$ can be expressed 
in terms of a density matrix $\rho_\Omega$ such that
\begin{equation}
\langle\Omega | \alpha |\Omega\rangle={\rm Tr}\big[ \rho_\Omega\, \alpha\big]\ ,\quad \forall\alpha\in{\cal C}_2\ ,
\label{8.13}
\end{equation}
with
\begin{equation}
\rho_\Omega=\frac{1}{2}\left( |g_1^{(1)}\rangle \langle g_1^{(1)} | +
|g_2^{(2)}\rangle \langle g_2^{(2)} | \right)\ ;
\label{8.14}
\end{equation}
the state $|\Omega\rangle$ is therefore a mixed state when restricted to ${\cal C}_2$.
In addition, being a convex combination of projectors onto separable states, $\rho_\Omega$
results itself separable, as already observed in Section 5, {\it cf.} Eq.(\ref{5.7-1}).
A similar conclusion holds also for the states $c_1 |\Omega\rangle$,
$c_2 |\Omega\rangle$ and $c_1 c_2|\Omega\rangle$ that together with $|\Omega\rangle$ generate
the whole GNS-Hilbert space ${\cal H}_\Omega$: one easily finds that, 
when restricted to the algebra ${\cal C}_2$,
also these three states are represented by separable density matrices.

More in general, any state on ${\cal C}_2$ can be represented by a density matrix that,
following (\ref{7.6}), can be written in the form:
\begin{equation}
\rho=\sum_{i,j}\, \lambda_{ij}\, | e_i\rangle\langle e_j|\ ,\qquad
\sum_i \lambda_{ii}=1\ ,
\label{8.15}
\end{equation}
where $\big\{ |e_i\rangle \, | \, i=1,2\big\}$ is any separable basis carrying a
two-dimensional representation of ${\cal C}_2$; in particular, one can choose
one of the two basis given in (\ref{8.3}) and (\ref{8.4}).

In order to characterize its entanglement properties, one has to distinguish the cases in which
the coefficient $\lambda_{12}$ is complex or real. In the first case,
one has:

\begin{lemma}
The density matrix $\rho$ as in (\ref{8.15}) with $\lambda_{12}$ a nonvanishing complex number
is never separable.
\end{lemma}

\begin{proof}
The density matrix in (\ref{8.15}) can be decomposed into its diagonal part,
\begin{equation}
\rho_D=\sum_{i}\, \lambda_{ii}\, | e_i\rangle\langle e_i|\ ,
\label{8.16}
\end{equation}
and the off-diagonal one $\eta\equiv\rho-\rho_D$. While $\rho_D$ is clearly a separable state,
$\eta$ being the difference of two density matrices is not even a state. However,
for $\lambda_{12}\in\mathbb{C}$, the quantity ${\rm Tr}[\eta\, c_1 c_2]
\equiv {\rm Tr}[\rho\, c_1 c_2]=2i\,{\cal I}m(\lambda_{12})$ is nonvanishing,
so that by the criterion of {\sl Corollary 3}, $\rho$ is surely entangled.
\end{proof}

When $\lambda_{ij}$ is a real matrix, the situation is more involved, since ${\rm Tr}[\rho\, c_1 c_2]$
is always zero and the entanglement criterion in {\sl Corollary 3} gives no information: one has then to resort to the fact that separable mixed state are convex combination of pure separable states.
\medskip

\begin{lemma}
The density matrix $\rho$ as in (\ref{8.15}) with $\lambda_{12}\in\mathbb{R}$ is separable if and only if
$\lambda_{11}\geq |\lambda_{12}|$ and $\lambda_{22}\geq |\lambda_{12}|$ 
\ .
\end{lemma}

\begin{proof}
Consider first a pure state $\vert\psi\rangle=a_1\,\vert e_1\rangle\,+\,a_2\,\vert e_2\rangle$, 
with $|a_1|^2+|a_2|^2=1$; one computes:
$$
\langle\psi\vert c_1\,c_2\vert\psi\rangle=2\,i\,{\cal I}m(a_1\,a_2^*)\ ,\quad
\langle\psi\vert c_1\vert\psi\rangle=|a_1|^2-|a_2|^2\ ,\quad
\langle\psi\vert c_2\vert\psi\rangle=2\,{\cal R}e(a_1\,a_2^*)\ .
$$
Using the separability condition (\ref{3.4}) of {\sl Lemma 1},
it follows that $\vert\psi\rangle$ is a separable pure state with respect to the considered bipartition if only if, 
together with ${\cal I}m(a_1a_2^*)=\,0$, at least one of the following two conditions is satisfied:
$|a_1|^2=|a_2|^2$ or ${\cal R}e(a_1\,a_2^*)=0$.
Therefore, the only separable pure states are:
$$
\vert\psi_1\rangle=\vert e_1\rangle\ ,\quad \vert\psi_2\rangle=\vert e_2\rangle\ ,\quad 
\vert\psi_3\rangle=\frac{\vert e_1\rangle\,+\,\vert e_2\rangle}{\sqrt{2}}\ ,\quad
\vert\psi_4\rangle=\frac{\vert e_1\rangle\,-\,\vert e_2\rangle}{\sqrt{2}}\ . 
$$
Take now a generic mixed, separable state $\rho$ that can be expressed as in (\ref{8.15});
it must be obtainable as a convex combination of the projectors onto the above separable pure states,
and therefore must be of the form:
\begin{equation}
\rho=
\sum_{i=1}^4 \mu_i\, |\psi_i\rangle\langle\psi_i|
\ ,\qquad \mu_i\geq 0\ ,\quad \sum_{i=1}^4\mu_i=1\ .
\label{8.16-1}
\end{equation}
As a consequence, comparing (\ref{8.15}) and (\ref{8.16-1}), one obtains: 
\begin{eqnarray*}
&&\lambda_{11}-\lambda_{12}=\mu_1+\mu_4\geq 0\ ,\qquad\hskip .5cm
\lambda_{11}+\lambda_{12}=\mu_1+\mu_3\geq 0\\ 
&&\lambda_{22}-\lambda_{12}=\mu_2+\mu_4\geq 0\ ,\hskip .6cm \qquad \lambda_{22}+\lambda_{12}=\mu_2+\mu_3\geq 0\ ,
\end{eqnarray*}
which are possible only if: $|\lambda_{12}|\leq {\rm min}\{\lambda_{11},\, \lambda_{22}\}$.

\noindent
On the other hand, using the above relations,
one can express three of the convex coefficients appearing in the decomposition (\ref{8.16-1})
in terms of the remaining one and $\lambda_{11}$, $\lambda_{12}$; for instance:
$$
\mu_2=1-2\lambda_{11}+\mu_1\ ,\quad \mu_3=\lambda_{11}+\lambda_{12}-\mu_1 ,\quad \mu_4=\lambda_{11}-\lambda_{12}-\mu_1\ .
$$
The conditions $0\leq \mu_i\leq 1$, $i=1,2,3$, then yields:
$$
\max\{\lambda_{11}-\lambda_{12}-1,\, \lambda_{11}+\lambda_{12}-1,\, 2\lambda_{11}-1\}\leq\,\mu_1\,\leq
\min\{\lambda_{11}-\lambda_{12},\, \lambda_{11}+\lambda_{12},\, 2\lambda_{11}\}\ .
$$ 
Therefore, assuming $|\lambda_{12}|\leq {\rm min}\{\lambda_{11},\, \lambda_{22}\}$, one can always choose
coefficients $\mu_i$, satisfying \hbox{$0\leq \mu_i\leq 1$}, $i=1,2,3,4$, 
such that a generic density matrix $\rho$ can be expressed 
as in (\ref{8.16-1}), {\it i.e.} as a convex combination of projectors on the separable pure states $|\psi_i\rangle$,
whence $\rho$ results itself separable.
½\end{proof}
\medskip

\noindent
{\bf Remark:}
Surprisingly, as explicitly shown in the above proof, the two combinations\break
\hbox{$|\psi_{3,4}\rangle=\big(|e_1\rangle \pm |e_2\rangle\big)/\sqrt2$} are
separable in ${\cal C}_2$. This behaviour is clearly quite different from the case 
of two distinguishable qubits, or two-mode boson/fermion systems, where superpositions of
of pure separable states give entangled states.
\hfill$\Box$

\medskip

In conclusion, using the powerful machinery of algebraic quantum mechanics
we have been able to classify all entangled states of the Clifford algebra ${\cal C}_2$.
As we shall see, one can treat in a similar way also the case of the general algebra ${\cal C}_N$.

\section{Structure of entangled Majorana states: ${\cal C}_N$}

Before treating the case of a general Clifford algebra, it is useful to explicitly
discuss the construction of a basis of separable states in ${\cal H}_\Omega$ carrying
the irreducible representations of ${\cal C}_N$ with $N=3,4$, by extending
the techniques previously adopted for ${\cal C}_2$.

The algebra ${\cal C}_3$, the simplest Clifford algebra ${\cal C}_N$ with $N$ odd,
 is the linear span of the set $\big\{ {\bf 1},\, c_1,\, c_2,\, c_3 \big\}$.
Its only non-trivial bipartition $({\cal A}_1, {\cal A}_2)$ is the one in which
${\cal A}_1$ is the linear span of $\big\{ {\bf 1},\, c_1,\, c_2 \big\}$, while ${\cal A}_2$
that of $\big\{ {\bf 1},\, c_3\big\}$, since all other possible bipartitions can be
reduced to this one by a suitable reordering of the mode labels.

Choosing again the state $\Omega$ as in (\ref{5.6}), the {\sl GNS construction} gives a
representation $\pi_\Omega$ of ${\cal C}_3$ on the Hilbert space ${\cal H}_\Omega$,
now eight-dimensional. As outlined earlier for ${\cal C}_2$, 
a separable basis in it can be constructed following the procedure
presented in the proof of {\sl Lemma~3}; in practice, such a basis can be obtained
by augmenting the four operators $\hat e_i^{(r)}(c_1,c_2)=(1\pm c_1)(1\pm c_2)$ 
introduced in (\ref{8.4}) 
with the additional two projectors 
$(1\pm c_3)/2$, yielding the eight orthonormal vectors:
\begin{equation}
|e_i^{(r,s)}\rangle= \frac{1}{2\sqrt2}\, \hat e_i^{(r)}(c_1,c_2)\, \hat e^{(s)}(c_3)\,
|\Omega\rangle\ ,\qquad r=1,2\ ,\quad i=1,2\ ,\quad s=1,2\ ,
\label{9.1}
\end{equation}
where
\begin{equation}
\hat e^{(1)}(c_3)=(1+c_3)\qquad
\hat e^{(2)}(c_3)=(1-c_3)\ .
\label{9.2}
\end{equation}

However, as discussed in Section 5, the irreducible representations of ${\cal C}_3$
are four-dimen-sional, so that $\pi_\Omega$ as given by the {\sl GNS construction}
decomposes as $\pi_\Omega=\pi_\Omega^{(1)}\oplus\pi_\Omega^{(2)}$
into two, equal, four-dimensional irreducible representations $\pi_\Omega^{(s)}$,
$s=1,2$ acting on two subspaces ${\cal H}_\Omega^{(s)}$, such that:
${\cal H}_\Omega={\cal H}_\Omega^{(1)} \oplus {\cal H}_\Omega^{(2)}$.
One can check that the first subspace ${\cal H}_\Omega^{(1)}$ is spanned by the
the four basis vectors in (\ref{9.1}) with $s=1$, 
{\it i.e.} $\big\{|e_i^{(r,1)}\rangle\ |\, r,i=1,2\big\}$, while the second
by those with $s=2$. Explicitly, one finds:

\begin{equation}
\begin{aligned}
&{\bf 1} \longrightarrow \sigma_0\otimes\sigma_0\\
&c_{1} \longrightarrow \sigma_0\otimes\sigma_3 \\
&c_{2} \longrightarrow \sigma_0\otimes\sigma_1 \\
&c_3 \longrightarrow \sigma_2\otimes\sigma_2\ .
\end{aligned}
\label{9.3}
\end{equation}
Therefore, the vectors in (\ref{9.1}) represent pure states for ${\cal C}_3$;
further, due to {\sl Proposition 1}, they are separable.

In order to explicit obtain the decomposition of the GNS- representation into its irreducible
components in the general case ${\cal C}_N$, with $N>3$, more effort is required.
The discussion of the case $N=4$ suffices to grasp the general structure.

In the case of ${\cal C}_4={\rm span}\big\{ {\bf 1},\, c_1,\, c_2,\, c_3,\, c_4 \big\}$,
the Hilbert space ${\cal H}_\Omega$ is 16-dimensional and can be spanned by the
sixteen orthonormal states:
\begin{equation}
\big|\, v_{(i,j)}^{(r,s)}\,\big\rangle= \frac{1}{4}\, \hat v_{(i,j)}^{(r,s)} \, |\Omega\rangle\ ,\qquad r,s=1,2\ ,\quad i,j=1,2\ ,
\label{9.4}
\end{equation}
with
\begin{equation}
\hat v_{(i,j)}^{(r,s)}\equiv\, \hat e_i^{(r)}(c_1,c_2)\ \hat e_j^{(s)}(c_3,c_4)\ ,
\label{9.5}
\end{equation}
where $\hat e_i^{(r)}(c_1,c_2)$ are the four combination in (\ref{8.4}),
while $\hat e_j^{(r)}(c_3,c_4)$ are exactly of the same form but with $c_1$
replaced by $c_3$, and $c_2$ by $c_4$. For instance, one explicitly has:
\begin{equation}
\begin{aligned}
&\hat v_{(1,1)}^{(1,1)}\equiv\, (1+c_1)(1+c_2)(1+c_3)(1+c_4)\\
&\hat v_{(2,1)}^{(1,1)}\equiv\, (1-c_1)(1+c_2)(1+c_3)(1+c_4)\\
&\hat v_{(1,2)}^{(1,1)}\equiv\, (1+c_1)(1+c_2)(1-c_3)(1+c_4)\\
&\hat v_{(2,2)}^{(1,1)}\equiv\, (1-c_1)(1+c_2)(1-c_3)(1+c_4)\ .
\end{aligned}
\label{9.6}
\end{equation}

The 16 states (\ref{9.4}) look separable for any choice of bipartition of ${\cal C}_4$, but unfortunately
can not be simply grouped into sets of four in order to form basis for subspaces of ${\cal H}_\Omega$ 
carrying irreducible representations of ${\cal C}_4$,
as done before for ${\cal C}_2$ and ${\cal C}_3$. Nevertheless, this can be obtained through
unitary transformations similar to the ones introduced in (\ref{8.8}) and (\ref{8.9});
this will allow to conclude that the vectors $\big|\, v_{(i,j)}^{(r,s)}\,\big\rangle$ are 
also pure states for ${\cal C}_4$. Recall that in general a state in ${\cal H}_\Omega$ results
mixed when restricted to the Clifford algebra, while in order to characterize
entangled Clifford states, a basis of pure, separable states is needed.
 
For sake of definiteness, let us fix the bipartition $({\cal A}_1, {\cal A}_2)$ of ${\cal C}_4$
for which ${\cal A}_1={\rm span}\big\{ {\bf 1},\, c_1,\, c_2\big\}$ and
${\cal A}_2={\rm span}\big\{ {\bf 1},\, c_3,\, c_4\big\}$.%
\footnote{ The other independent bipartition, for which
${\cal A}_1={\rm span}\big\{ {\bf 1},\, c_1,\, c_2,\, c_3 \big\}$
and ${\cal A}_2={\rm span}\big\{ {\bf 1},\, c_4 \big\}$, can be similarly treated
using the results obtained above for ${\cal C}_3$.}
Using the unitary matrices $U$ and $V$ in (\ref{8.9}), one can then write:
\begin{equation}
\hat v_{(i,j)}^{(r,s)}=\sum_{p,k} U_{rp}\, V_{ik}\, \hat f_k^{(p)}(c_1,c_2)
\sum_{q,\ell} U_{sq}\, V_{j\ell}\, \hat f_\ell^{(q)}(c_3,c_4)\ ,
\label{9.7}
\end{equation}
where $\hat f_k^{(p)}(c_1,c_2)$, $p,k=1,2$, coincide with the monomials in (\ref{8.11}),
while $\hat f_\ell^{(q)}(c_3,c_4)$, $q,\ell=1,2$, are exactly of the same form with
the substitution $c_1\rightarrow c_3$ and $c_2\rightarrow c_4$. This implies that the states
$\big|\, v_{(i,j)}^{(r,s)}\,\big\rangle$ in (\ref{9.4}) can be expressed as linear
combinations of the vectors:
\begin{equation}
\big|\, f_{(i,j)}^{(r,s)}\,\big\rangle= \frac{1}{4}\, \hat f_i^{(r)}(c_1,c_2)\
\hat f_j^{(s)}(c_3,c_4) \, |\Omega\rangle\ ,\qquad r,s=1,2\ ,\quad i,j=1,2\ .
\label{9.8}
\end{equation}
Using the results presented in the previous Section, one easily checks
that the four vectors $\big\{ \big|\, f_{(i,j)}^{(r,s)}\,\big\rangle\ |\, i,j=1,2\big\}$
with the indices $r$ and $s$ fixed, span a four-dimensional subspace 
${\cal H}_\Omega^{(r,s)}\subset{\cal H}_\Omega$ carrying an irreducible representation
$\pi_\Omega^{(r,s)}$ of ${\cal C}_4$. Then, the original GNS-representation $\pi_\Omega$
decomposes as 
\begin{equation}
\pi_\Omega=\oplus_{r,s}\ \pi_\Omega^{(r,s)}\ ,
\label{9.10}
\end{equation}
into four, 4-dimensional, irreducible representations $\pi_\Omega^{(r,s)}$, that turn out to be
all equal. As a result, since the basis vectors $\big|\, f_{(i,j)}^{(r,s)}\,\big\rangle$ are
pure, the original vectors $\big|\, v_{(i,j)}^{(r,s)}\,\big\rangle$, being linear combinations
of these, are also pure. In addition, they are also separable, being essentially product states.

\medskip
\noindent
{\bf Remark:} Notice that the states $\big|\, f_{(i,j)}^{(r,s)}\,\big\rangle$ 
are manifestly separable for the chosen bipartition $({\cal A}_1, {\cal A}_2)$; however,
as discussed in the previous Section, these states result non-separable when restricted
to the two-dimensional Clifford subalgebras ${\cal A}_1$ or ${\cal A}_2$.
\hfill$\Box$
\medskip

The whole construction can now be easily generalized to the case of a generic
Clifford algebra ${\cal C}_N$. 
Given the state $\Omega$ in (\ref{5.6}), when $N$ is even one can
build a basis of pure, separable states in the Hilbert space ${\cal H}_\Omega$
by acting with products of the four elements $\hat e_i^{(r)}(c_a,c_b)=(1\pm c_a)(1\pm c_b)$
on the cyclic GNS-vector $|\Omega\rangle$,
explicitly obtaining:
\begin{equation}
\big|\, e_{\bf i}^{\bf r}\,\big\rangle
=\frac{1}{2^{N/2}}\ \hat e_{i_1}^{(r_1)}(c_1,c_2) \, \hat e_{i_3}^{(r_3)}(c_3,c_4)\ldots
\hat e_{i_{N-1}}^{(r_{N-1})}(c_{N-1},c_N)\ |\Omega\rangle\ ,
\label{9.11}
\end{equation}
where ${\bf r}=(r_1, r_3,\ldots, r_{N-1})$ and ${\bf i}=(i_1, i_3,\ldots, i_{N-1})$;
on the other hand, when $N$ is odd, also the two elements $\hat e_i^{(r)}(c_a)=(1\pm c_a)$ are needed,
so that:
\begin{equation}
\big|\, e_{\bf i}^{\bf r}\,\big\rangle
=\frac{1}{2^{N/2}}\ \hat e_{i_1}^{(r_1)}(c_1,c_2) \, \hat e_{i_3}^{(r_3)}(c_3,c_4)\ldots
\hat e_{i_{N-1}}^{(r_{N-1})}(c_{N-2},c_{N-1})\, \hat e_{i_N}^{(r_N)}(c_N)\ |\Omega\rangle\ .
\label{9.12}
\end{equation}
In the above expressions, all indices $r$'s and $i$'s take the two values 1 and 2.
These states are all manifestly separable for any bipartition of the algebra ${\cal C}_N$; 
further they are pure, since, as in the case of ${\cal C}_4$ discussed above, they can be
unitarily related to states carrying irreducible representations of the Clifford algebra.

For sake of definiteness, let us assume $N$ to be even and 
fix a bipartition $({\cal A}_1, {\cal A}_2)$
for which ${\cal A}_1={\rm span}\big\{ {\bf 1},\, c_1,\, c_2\, \ldots, c_{2k}\big\}$ and
${\cal A}_2={\rm span}\big\{ {\bf 1},\, c_{2k+1},\, c_{2k+2},\ldots, c_N\big\}$, with
\hbox{$1 \leq k \leq N/2-2$}; in this case, the multi-indices $\bf r$ and $\bf i$
take $2^{N/2}$ possible values, that we shall henceforth take to be: $1,2,\ldots, 2^{N/2}$.
Then, by generalizing the transformation in (\ref{9.7}),
the states $\big|\, e_{\bf i}^{\bf r}\,\big\rangle$
in (\ref{9.11}) can be unitarily related to the following ones:
\begin{equation}
\big|\, f_{\bf i}^{\bf r}\,\big\rangle
=\frac{1}{2^{N/2}}\ \hat f_{i_1}^{(r_1)}(c_1,c_2) \, \hat f_{i_3}^{(r_3)}(c_3,c_4)\ldots
\hat f_{i_{N-1}}^{(r_{N-1})}(c_{N-1},c_N)\ |\Omega\rangle\ ,
\label{9.13}
\end{equation}
with $\hat f_i^{(r)}(c_a,c_b)=\{(c_a\pm c_b),\ (1\pm i c_a c_b)\}$, as in (\ref{8.11}).%
\footnote{
Similarly, for the independent bipartition with
${\cal A}_1={\rm span}\big\{ {\bf 1},\, c_1,\, c_2\, \ldots, c_{2k+1}\big\}$ and
${\cal A}_2={\rm span}\big\{ {\bf 1},\, c_{2k+2},\, c_{2k+3},\ldots, c_N\big\}$,
the states in (\ref{9.11}) can be unitarily related to the following (unnormalized) ones:
$\hat f_{i_1}^{(r_1)}(c_1,c_2)\ldots \, \hat f_{i_{2k-1}}^{(r_{2k-1})}(c_{2k-1},c_{2k})\,
\hat e_{i_{2k+1}}^{(r_{2k+1})}(c_{2k+1})\, \hat f_{i_{2k+2}}^{(r_{2k+2})}(c_{2k+2},c_{2k+3})
\ldots \hat e_{i_N}^{(r_N)}(c_N) |\Omega\rangle$.}
For fixed indices ${\bf r}$, these states 
span a $2^{N/2}$-dimensional subspace ${\cal H}_\Omega^{({\bf r})}$ of ${\cal H}_\Omega$ 
carrying an irreducible representation of ${\cal C}_N$. 
Indeed, the $2^{N/2}$ sets of basis vectors 
$\big\{ | f_{\bf i}^{({\bf r})}\rangle\ |\, {\bf i}=1,2,\ldots, 2^{N/2}\big\}$, 
with ${\bf r}=1,2,\ldots,2^{N/2}$, 
induce a decomposition ${\cal H}_\Omega=\oplus_{\bf r}\ {\cal H}_\Omega^{({\bf r})}$ of the GNS-Hilbert space
into subspaces, each carrying an irreducible representation $\pi_\Omega^{({\bf r})}$ of ${\cal C}_{N}$,
so that the GNS-representation $\pi_\Omega$ has the following decomposition into irreducible components:
$\pi_\Omega=\oplus_{\bf r}\ \pi_\Omega^{({\bf r})}$.

Since all the representation $\pi_\Omega^{({\bf r})}$ turn out to be equal, 
any state of ${\cal C}_{N}$ in the {\sl folium} of $\Omega$,
represented by a density matrix $\rho$, can be decomposed as
\begin{equation}
\rho=\sum_{{\bf j},{\bf k}}\, \lambda_{{\bf j}{\bf k}}\, | f_{\bf j}\rangle\langle f_{\bf k}|\ ,\qquad
\sum_{\bf k} \lambda_{{\bf k}{\bf k}}=1\ ,
\label{9.14}
\end{equation}
where for the vector basis $\big\{| f_{\bf k}\rangle \big\}$ one can choose any of the above introduced sets
$\big\{| f_{\bf k}^{({\bf r})}\rangle \big\}$, with $\bf r$ fixed. 
Clearly, also in this general case, a diagonal state of the form
\begin{equation}
\rho_D=\sum_{{\bf k}}\, \lambda_{{\bf k}{\bf k}}\, | f_{\bf k}\rangle\langle f_{\bf k}|\ ,
\label{9.15}
\end{equation}
is manifestly separable. On the other hand, provided not all off-diagonal coefficients $\lambda_{{\bf j}{\bf k}}$,
${\bf j}\neq {\bf k}$, are real, one can always find a monomial $c_{i_1} c_{i_2}$ with
$c_{i_1}\in {\cal A}_1$ and $c_{i_2}\in {\cal A}_2$
for which ${\rm Tr}[\eta\, c_{i_1} c_{i_2}]\neq0$, with $\eta =\rho-\rho_D$. 
Therefore, in this generic case, by the criterion of {\sl Corollary~3} any state $\rho$
results entangled if and only if it is not in the diagonal form (\ref{9.15}). However,
a full characterization of entangled states in the case in which
the coefficients $\lambda_{{\bf j}{\bf k}}$ are all real
can not be given in general, since one has to resort to the general
separability condition (\ref{3.1}).

As an interesting example of an entangled state in ${\cal C}_N$, let us fix $N=2n$
and consider the balanced bipartition $({\cal A}_1, {\cal A}_2)$
for which ${\cal A}_1={\rm span}\big\{ {\bf 1},\, c_1,\, c_2\, \ldots, c_n\big\}$ and
${\cal A}_2={\rm span}\big\{ {\bf 1},\, c_{n+1},\, c_{n+2},\ldots, c_{2n}\big\}$.
The monomial 
\begin{equation}
\gamma=\gamma_p^{(1)}\, \gamma_p^{(2)}\ , \qquad \gamma^{(1)}=c_{i_1} c_{i_2}\ldots c_{i_p}\ ,
\quad \gamma^{(2)}=c_{j_1} c_{j_2}\ldots c_{j_p}\ ,
\label{9.16}
\end{equation}
with $1\leq i_k\leq n$ and $(n+1)\leq j_k\leq 2n$,
is an element of ${\cal C}_{2n}$ which is manifestly local with respect to the chosen bipartition,
since $\gamma_p^{(1)}\in {\cal A}_1$, while $\gamma_p^{(2)}\in {\cal A}_2$.
Furthermore, when the integer $p$ is odd, $\gamma^{(1)}$ and $\gamma^{(2)}$ are odd elements,
$\big\{\gamma^{(1)},\, \gamma^{(2)}\big\}=0$, such that $\gamma^2=-1$. Consider then the following
vector in ${\cal H}_\Omega$:
\begin{equation}
|\phi\rangle= \frac{1}{\sqrt2}\big(1+i\gamma\big)\,|\Omega\rangle\ .
\label{9.17}
\end{equation}
When restricted to the algebra ${\cal C}_{2n}$ it becomes a mixed state, since
both $|\Omega\rangle$ and $\gamma |\Omega\rangle$ are no longer pure; further,
the expectation $\langle \phi|\gamma|\phi\rangle$ is non vanishing, so that, again by 
{\sl Corollary 3} the state is entangled. This result will be useful in the following
Section, while discussing metrological applications of Majorana systems.

\section{Application to quantum metrology}

Using quantum physics in metrological applications
is surely one of the most promising developments in quantum technology:
it allows determining a physically interesting parameter 
$\theta$, typically a phase, with unprecedented accuracy.%
\footnote{The literature on the subject is fast growing; for a partial list, 
see \cite{Caves1}-\cite{Giovannetti2} and references therein.}
This result is achieved through a $\theta$-dependent
state transformation that occurs inside a suitable measurement apparatus, generally
an interferometric device.
In the most common case of linear setups, this transformation
can be modelled by a unitary mapping, $\rho \to \rho_\theta$, sending the
initial state $\rho$ into the final parameter-dependent outcome state:
\begin{equation}
\rho_\theta= e^{i\theta J}\, \rho\, e^{-i\theta J}\ ,
\label{10.1}
\end{equation}
where $J$ is the devices-dependent, $\theta$-independent operator generating
the state transformation.
The task of quantum metrology is to determine the ultimate bounds on the accuracy with which
the parameter $\theta$ can be obtained through a measurement of $\rho_\theta$
and to study how these bounds scale with the available resources.

General quantum estimation theory allows a precise determination of
the accuracy $\delta\theta$ with which the phase $\theta$ can be obtained in a measurement involving
the operator $J$ and the initial state $\rho$; one finds that $\delta\theta$ is limited by the
following inequality \cite{Helstrom}-\cite{Paris}:
\begin{equation}
\delta\theta\geq {1\over \sqrt{F[\rho, J]}}\ ,
\label{10.2}
\end{equation}
where the quantity $F[\rho, J]$ is the so-called ``quantum Fisher information'';
it is a continuous, convex function
of the state $\rho$, satisfying the inequality \cite{Luo,Braunstein}
\begin{equation}
F[\rho, J]\leq 4\, \Delta^2_{\rho} J\ ,
\label{10.3}
\end{equation}
where $\Delta^2_{\rho} J\equiv\big[ \langle J^2\,\rangle-\langle J\,\rangle^2\big]$ 
is the variance of the operator $J$ in the state $\rho$, the equality holding
only for pure states.
In order to reach a better resolution in $\theta$-estimation 
one should obtain larger quantum Fisher information; thus,
for a given measuring apparatus, {\it i.e.} a given operator $J$, one can still optimize
the precision with which $\theta$ is determined by choosing an initial state $\rho$
that maximizes $F[\rho, J]$.

For measuring devices made of $N$ {\sl distinguishable} particles,
the following bound on the quantum Fisher information holds
for any separable state $\rho_{\rm sep}$ \cite{Smerzi}:
\begin{equation}
F\big[\rho_{\rm sep}, J\big]\leq N\ .
\label{10.4}
\end{equation}
In other terms, feeding the measuring apparatus with separable initial states, the best 
achievable precision in the determination of the phase shift $\theta$ is bounded
by the so-called shot-noise limit:
\begin{equation}
\delta\theta\geq{1\over\sqrt{N}}\ .
\label{10.5}
\end{equation}
This is also the best result
attainable using classical ({\it i.e.} non quantum) devices:
the accuracy in the estimation of $\theta$ scales at most with the inverse
square root of the number of available resources.
Instead, quantum equipped metrology allows reaching sub-shot-noise sensitivities
by using suitable detection protocols and entangled input states.

This conclusion holds when the metrological devices used to estimate the physical
parameter $\theta$ are based on systems of distinguishable particles. When dealing with
{\sl identical} particles, the above statement needs to be rephrased.
Indeed, both in the case of boson and fermion
systems it has been explicitly shown that sub shot-noise
sensitivities may be obtained also via a non-local
operation acting on separable input states \cite{Benatti1,Benatti6}. 
In other terms, although some sort of non-locality is needed in order to go below the
shot-noise limit, this can be provided by the measuring apparatus itself and
not necessarily by the input state $\rho$, that indeed can be separable. This result
has clearly direct experimental relevance, since the preparation of
suitably entangled input states may require in practice a large amount of resources.

When dealing with fermions, the situation appears more involved due to the
anticommuting character of the associated operator algebra $\cal A$.
Indeed, while in the case of bosons a two-mode apparatus ({\it e.g.} a standard
two-way interferometer) filled with $N$ particles is sufficient to reach sub-shot-noise
efficiencies, with fermions a multimode apparatus is needed in order to reach
comparable sensitivities \cite{Dariano1}-\cite{Cooper2}.

With Majorana fermions, things become even more complicated since the notion of 
mode occupation loses its meaning, as a number operator is no longer available.
It is the number $N$ of available Majorana modes that now quantifies the amount of resources
available for the process of parameter estimation and it is in reference to this number
that the shot-noise limit in (\ref{10.5}) should be considered. In other terms,
in dealing with systems of Majorana fermions, it is the mode structure of the 
measuring apparatus that becomes relevant.

As an example, consider again a system with an even number $N=2n$ of Majorana modes and choose for
it the balanced bipartition $({\cal A}_1, {\cal A}_2)$, with 
${\cal A}_1={\rm span}\big\{ {\bf 1},\, c_1,\, c_2\, \ldots, c_n\big\}$ and
${\cal A}_2={\rm span}\big\{ {\bf 1},\, c_{n+1},\, c_{n+2},\ldots, c_{2n}\big\}$.
As a generator of the transformations inside
the measuring apparatus, let us take the following operator:
\begin{equation}
J=i\sum_{k=1}^n \omega_k\, c_k c_{n+k}\ ,
\label{10.6}
\end{equation}
where $\omega_k$ is a given spectral function, {\it e.g.} $\omega_k\simeq k^p$,
with $p$ integer. The unitary transformation $U_\theta=e^{i\theta J}$ implementing
the state transformation inside the apparatus is clearly non-local,
since it can not be written as the product $\alpha^{(1)}\,\alpha^{(2)}$,
with $\alpha^{(1)}\in {\cal A}_1$ and $\alpha^{(2)}\in {\cal A}_2$.
It represents a sort of generalized multimode beam-splitter, so that the whole measuring device
behaves as a multimode interferometer.

Let us feed the interferometer with a pure initial state $\rho=|\psi\rangle\langle\psi|$,
coinciding with one of the basis elements $\big|\, f_{\bf i}^{\bf r}\,\big\rangle$
in (\ref{9.13})
carrying the irreducible representation $\pi_\Omega^{({\bf r})}$ of ${\cal C}_{2n}$
presented in the previous Section. For instance, choose:
\begin{equation}
|\psi\rangle=\frac{1}{2^n} (1+ic_1 c_2)\ldots (1+ic_{n-1} c_n)\, (1-ic_{n+1} c_{n+2})\ldots
(1-ic_{2n-1} c_{2n})\ |\Omega\rangle\ ;
\label{10.7}
\end{equation}
as discussed before, this state is separable with respect to the chosen bipartition.

The quantum Fisher information can be easily computed, since in this case it is proportional
to the variance of $J$ with respect to $|\psi\rangle$; assuming for simplicity $n$ even, one finds
\begin{equation}
F[\rho, J]=\sum_{k=1} ^{N/4} \big( \omega_{2k-1} + \omega_{2k} \big)^2\ ,
\label{10.8}
\end{equation}
which turns out to be larger than $N$. In particular, for $\omega_k\sim k$,
one finds that in the limit of large $N$, $F[\rho, J]$ behaves as $N^3/3$.
Therefore, also with Majorana fermions, a suitably devised interferometric apparatus can beat the
shot-noise limit in $\theta$ estimation even starting with a separable state as in (\ref{10.7}).
One can check that this sub-shot-noise gain in precision can be obtained also
using any other state vector belonging to the separable basis (\ref{9.11}), although 
it is for the state (\ref{10.7}) that the
value actually attained by the quantum Fisher information is maximal.

\medskip
\noindent
{\bf Remark:} 
Notice that in general the obtained value for $F[\rho, J]$ scales with a power of $N$
greater than two. When dealing with systems made of $N$ distinguishable particles,
the following general bound holds:
\begin{equation}
F[\rho, J]\leq N^2\ ,
\label{10.9}
\end{equation}
for any state $\rho$ and generator $J$, providing an absolute lower bound for the accuracy
in $\theta$ estimation called the Heisenberg limit: $\delta\theta\geq 1/N$. 
Instead, in the scenario described above, one is able to reach sub-Heisenberg sensitivities, another
advantage of using fermion systems.%
\footnote{The possibility of getting sensitivities beyond the Heisenberg limit
has been discussed before, using, however, non-linear metrology \cite{Luis}-\cite{Hall}.}
\hfill$\Box$

\medskip
Some sort of quantum non-locality is nevertheless needed in order to attain sub-shot-noise accuracies.
In order to appreciate this point, let us consider the same Majorana system as before, 
but use a different generator $\tilde J$, {\it i.e.} a different measuring apparatus,
where:
\begin{equation}
\tilde J=i\sum_{k=1}^n \tilde\omega_k\, c_{2k-1} c_{2k}\ ,
\label{10.10}
\end{equation}
with $\tilde\omega_k$ a given spectral density. The unitary transformation 
$\tilde U_\theta=e^{i\theta\tilde J}$ implementing state transformation inside the interferometer 
is now local with respect to the chosen bipartition, since it is the product of $n$ transformations depending on
couples of contiguous modes:
\begin{equation}
\tilde U_\theta=\prod_{k=1}^n e^{-\theta\, \tilde\omega_k c_{2k-1} c_{2k}} \ .
\label{10.11}
\end{equation}
If one feeds the apparatus with any vector belonging to the separable basis in (\ref{9.11}),
one does not obtain any advantage in parameter estimation accuracy with respect
to the shot-noise limit; actually, the quantum Fisher information vanishes.

However, using an entangled state as initial state, the situation changes. 
Indeed, let us consider the entangled state $|\phi\rangle$ in (\ref{9.17}) introduced
at the end of the previous Section. Although not a pure state for ${\cal C}_N$,
the corresponding quantum Fisher information, 
being representation independent, can be computed in the full Hilbert space ${\cal H}_\Omega$,
obtaining
\begin{equation}
F\big[|\phi\rangle,\, \tilde J\,\big]=\sum_{k=1} ^{n} \big(\tilde\omega_k\big)^2\ .
\label{10.12}
\end{equation}
For a spectral density of generic form $\tilde\omega_k\simeq k^p$,
one finds, for large $N$,  $F\big[|\phi\rangle,\, \tilde J\,\big]\simeq N^{2p+1}$,
obtaining again a sub-shot-noise accuracy for $\theta$ estimation; actually,
for $p\geq1$, the sensitivity in the determination of $\theta$ goes even beyond the
Heisenberg limit.

\section{Outlook}

Non-classical correlations are at the basis of most of the recent advances
in modern quantum physics, and in particular in quantum technology,
leading to the possibility of the realization of devices 
outperforming those presently available.
The characterization and quantification of these resources is therefore of utmost
importance, especially in many-body systems, since,
thanks to the recent advances in ultracold and superconducting physics, 
they are becoming the preferred laboratories for studying new quantum effects.

For systems made of identical constituents, the usual notions of separability
and entanglement need a revision, since the particle Hilbert space tensor structure
on which these concepts are based is no longer available due to particle indistinguishability.
The attention should then shift from the Hilbert space paradigm
to a new one, focusing on the system observables and the algebra they obey; quantum
non-separability can then be signaled by the behaviour of observable correlation
functions.

This change of perspective can be most simply formulated using
the algebraic approach to quantum physics. There, a quantum system is identified
by its operator algebra $\cal A$ containing all its observables, while the Hilbert space ${\cal H}_\Omega$
of its states is an emergent concept, determined by the choice of a state $\Omega$ on $\cal A$
through the so-called {\sl GNS construction}. The state $\Omega$, a positive, normalized linear form on $\cal A$, 
determines the expectation values of the observables, thus making the connection with measurable quantities.
It also provides an explicit representation $\pi_\Omega$ of $\cal A$,
so that the observables act as operators on ${\cal H}_\Omega$.
In this framework, the notion of locality is no longer
given {\it a priori}, once for all; rather, it is based on the choice of a bipartition 
(or more in general multipartition) of the algebra $\cal A$
into subalgebras  ${\cal A}_1$ and ${\cal A}_2$, such that
${\cal A}_1 \cup {\cal A}_2={\cal A}$ and ${\cal A}_1 \cap {\cal A}_2={\bf 1}_{\cal A}$.
An element of $\cal A$ is local if it is the product of an element of ${\cal A}_1$
times an element of ${\cal A}_2$.
A state $\Omega$ on $\cal A$ is then separable if its expectation on all local
operators can be written as a convex combinations of products of expectations.

This general definition of separability, previously studied in boson
or fermion settings, has been here applied to the study of systems made of Majorana fermions.
In view of the attention they are receiving in superconducting physics
and as possible building blocks in topological quantum computations,
Majorana excitations are becoming the focus of a rapidly increasing number of investigations:
studying their entanglement properties is therefore of
great relevance.

For Majorana systems, the operator algebra $\cal A$ containing all observables 
results a Clifford algebra $\cal C$. These algebras do not admit a Fock representation;
this implies that for such systems the notions of number operator 
and of mode occupation are no longer available. Furthermore, given a state
$\Omega$, the corresponding representation of $\cal C$ on the Hilbert space ${\cal H}_\Omega$ turns out
to be in general reducible: this makes the characterization of entangled states
much more involved than for bosons or ordinary fermions, since now $\Omega$
is no longer a pure state for the algebra $\cal C$.

The relation between quantum non-separability
and reducibility of the GNS-representation $\pi_\Omega$ 
has not been much studied in the literature.
Here instead, a general, detailed treatment of entanglement theory
in presence of reducible operator algebra representations has been given, 
and then applied to the
case of Clifford algebras for a specific, physically relevant
choice of the state $\Omega$. This has allowed obtaining
a rather complete characterization
of general entangled Majorana states; as an illustration,
the cases of systems containing just few Majorana modes
have been analyzed in great detail.
The whole treatment is very general and can be easily applied to
discuss different choices for $\Omega$.

Among promising quantum technological applications, quantum metrology is the natural
context in which the above results can be fruitfully employed.
Indeed, as discussed in the last Section, multimode Majorana quantum interferometers
can be used to improve the accuracy in the measurement of relevant physical parameters
much beyond the so called shot-noise limit, the best limit reachable by classical devices.
Some sort of quantum non-locality is clearly needed in order to reach 
these sub-classical sensitivities;
however, this need not be encoded in the initial state, it can be provided by the
interferometric apparatus itself. As a result, no preliminary, resource consuming,
entanglement operation on the state entering the apparatus is needed in order
to get sub-shot-noise accuracies. In this respect, Majorana fermion systems may turn out
to play a central role in the development of new generations of
quantum sensors capable of outperforming any available apparatus dedicated
to the detection of faint physical signals.

\vfill\eject

\end{document}